\documentclass[letterpaper, 10 pt, journal]{ieeetran}
\usepackage{amsmath}

\usepackage{amsthm}
\usepackage{amsfonts}
\usepackage{bbm} 
\usepackage{graphicx}
\usepackage{subfigure}
\usepackage{cite}
\usepackage{algorithm}
\usepackage{algpseudocode}
\usepackage{nicefrac}
\usepackage{eufrak}
\usepackage{upgreek}
\usepackage{mathrsfs}
\usepackage{enumitem}

\usepackage{xcolor}
\definecolor{forestgreen}{rgb}{0.13, 0.55, 0.13}
\definecolor{orange}{rgb}{1,0.49,0}

\newtheorem{defn}{Definition}
\newtheorem{thm}{Theorem}
\newtheorem{lem}{Lemma}
\newtheorem{prop}{Proposition}
\newtheorem{cor}{Corollary}
\newtheorem{note}{Note}
\newtheorem{rem}{Remark}
\newtheorem{exam}{Example}
\newtheorem{assump}{Assumption}
\newtheorem{prob}{Problem}

\newcommand{\e}{\varepsilon}
\newcommand{\G}{\mathcal{G}}
\newcommand{\V}{\mathbb{V}}
\newcommand{\E}{\mathcal{E}}
\newcommand{\EE}{\mathbb{E}}

\newcommand{\R}{\mathbb{R}}
\newcommand{\PP}{\mathbb{P}}

\newcommand{\HH}{\mathcal{H}}
\newcommand{\N}{\mathcal{N}}

\usepackage[color=yellow, textsize=small, textwidth=\marginparwidth, draft]{todonotes}   

\begin{document}
\title{{Monitoring Link Faults in Nonlinear Diffusively-coupled Networks}}
\author{Miel Sharf,~\IEEEmembership{Graduate~Student~Member,~IEEE} and Daniel Zelazo,~\IEEEmembership{Senior~Member,~IEEE}
\thanks{M. Sharf and D. Zelazo are with the Faculty of Aerospace Engineering, Israel Institute of Technology, Haifa, Israel.
    {\tt\small mielsharf@gmail.com, dzelazo@technion.ac.il}.  }
}

\maketitle
\begin{abstract}
Fault detection and isolation is {an area} of engineering dealing with designing on-line protocols for systems that allow one to identify the existence of faults, pinpoint their exact location, and overcome them. We consider the case of multi-agent systems, where faults correspond to the disappearance of links in the underlying graph, simulating a communication failure between the corresponding agents. We study the case in which the agents and controllers are maximal equilibrium-independent passive (MEIP), and use the known connection between steady-states of these multi-agent systems and network optimization theory. We first study asymptotic methods of differentiating the faultless system from its faulty versions by studying their steady-state outputs. We explain how to apply the asymptotic differentiation to {detect and isolate communication faults}, with graph-theoretic guarantees on the number of faults that can be isolated, assuming the existence of a ``convergence assertion protocol", a data-driven method of asserting that a multi-agent system converges to a conjectured limit. We then construct two data-driven model-based convergence assertion protocols. We demonstrate our results by a case study.
\end{abstract}

\section{Introduction}\label{sec.Intro}
Multi-agent systems (MAS) have been widely studied in recent years, as they present both a variety of applications, as well as a deep theoretical framework \cite{Oh2015,Scardovi2010,OlfatiSaber2007}. One of the deepest concerns when considering applications of MAS is communication failures, which can drive the agents to act poorly, or fail their task altogether. These communication failures, {which we term \emph{network failures}}, can either be accidental or planned by an adversary. There is a need of detecting network faults and dealing with them in real-time for the network to be secure.

{{Fault detection and isolation (FDI) for multi-agent systems usually deals with faults in one of the agents, see e.g. \cite{Teixeira_Shames_Sandberg_Johansson} and references therein. The possibility of faults in the communication links was first studied in \cite{Jafai2010} using the notion of structural controllability, which was later used in \cite{Jafari2011} to solve the problem of leader localization. The problem of network FDI, {i.e., detection and isolation of network faults,} was  studied primarily for linear and time-invariant (LTI) systems with a known model. In \cite{Rahimian2014a,Rahimian2015}, the authors use jump discontinuities in the derivative of the output to detect topological changes in the network. Tools from switching systems theory, namely mode-observability, was used in \cite{Battistelli2015} for network FDI. Combinatorial tools were used in \cite{Rahimian2012,Valcher2019} to solve the FDI problem for consensus-seeking networks. Recently, \cite{Zhang2020} proposed a network FDI method which allows an uncertainty in the model, but is restricted for networks with LTI systems. A related problem in which one tries to distinguish between multi-agent systems with identical agents but different communication graphs was studied in \cite{Rahimian2014b,Sharf2018b,Patil2019}, from which only \cite{Sharf2018b} also deals with nonlinear agents. 

We aim at a network FDI scheme applicable also for nonlinear systems by relying on another concept widespread in multi-agent systems, namely passivity. Passivity was first used to address faults by \cite{Yang2008} for control-affine systems, although only fault-tolerance is addressed, and no synthesis procedures are suggested. Later works used FDI for a single nonlinear agent \cite{Yang2008,Chen2010,Marton2012b,Lei2018}. To the extent of our knowledge, passivity has not been previously used to give network FDI schemes, and no other works consider networks with nonlinear components.}}

Passivity theory is a cornerstone of the theoretical framework of networks of dynamical systems \cite{Bai_Arcak_Wen}. It allows for the analysis of multi-agent systems to be decoupled into two separate layers, the dynamic system layer and the information exchange layer. Passivity theory was first used to study the convergence properties of network systems in \cite{Arcak}. Many variations and extensions of passivity have been applied in different aspects of multi-agent systems. For example, the related concepts of incremental passivity or relaxed co-coercivity have been used to study various synchronization problems \cite{Stan2007, Scardovi2010}, and more general frameworks including Port-Hamiltonian systems on graphs \cite{Schaft2012}. Passivity is also widely used in coordinated control of robotic systems \cite{Chopra2006}.

One prominent variant is maximal equilibrium-independent passivity (MEIP), which was applied in \cite{SISO_Paper} in order to reinterpret the analysis problem for a multi-agent system as a pair of network optimization problems. Network optimization is a branch of optimization theory dealing with optimization of functions defined over graphs \cite{Rockafellar1998}. The main result of \cite{SISO_Paper} showed that the asymptotic behavior of these networked systems is \emph{(inverse) optimal} with respect to a family of network optimization problems. In fact, the steady-state input-output signals of both the dynamical systems and the controllers comprising the networked system can be associated to the optimization variables of either an \emph{optimal flow} or an \emph{optimal potential} problem; these are the two canonical dual network optimization problems described in \cite{Rockafellar1998}. The results of \cite{SISO_Paper} were used in \cite{LCSS_Paper,TAC_Paper} in order to solve the synthesis problem for multi-agent systems, and were further used in \cite{Sharf2018b,Sharf2019d} to solve the network identification problem.

We aim to use the network optimization framework of \cite{SISO_Paper,LCSS_Paper,TAC_Paper} for analysis and synthesis of multi-agent systems in order to provide a {strategy for detecting and isolating network faults}. We also consider adversarial games regarding communication faults. We strive to give graph-theoretic-based results, showing that {network} fault detection and isolation can be done for any MEIP multi-agent system, so long that the graph $\G$ satisfies certain conditions. We show that if the graph $\G$ is ``connected enough", we can solve the network FDI problem. Namely, if $\G$ is $2$-connected, then detecting the existence of any number of faults is possible, and if $\G$ is $k$-connected with $k>2$, we can isolate up to $k-2$ faults.

The rest of the paper is as follows. Section \ref{sec.background} surveys the relevant parts of the network optimization framework. Section \ref{sec.probform} presents the problem formulation of this work and states the assumptions used throughout the paper.  Section \ref{sec.IndicationVectors} presents the first technical tool used for building the {network} fault detection schemes, namely edge-indication vectors, and shows how to construct them. Section \ref{sec.Applications} uses edge-indication vectors to design network FDI schemes, as well as strategies for adversarial games, assuming the existence of a ``convergence assertion protocol", a data-driven method of asserting that a given MAS converges to a conjectured limit. Section \ref{sec.ConvAssert} studies these convergence assertion protocols, prescribing two approaches for constructing them. Lastly, we present simulations demonstrating the constructed algorithms.

\paragraph*{Notations}
We use basic notions from algebraic graph theory \cite{Godsil_Royle}. An undirected graph $\mathcal{G}=(\mathbb{V},\mathbb{E})$ consists of a finite set of vertices $\mathbb{V}$ and edges $\mathbb{E} \subset \mathbb{V} \times \mathbb{V}$.  We denote by $e=\{i,j\} \in \mathbb{E}$ the edge that has ends $i$ and $j$ in $\mathbb{V}$. For each edge $e$, we pick an arbitrary orientation and denote $e=(i,j)$.  The incidence matrix of $\mathcal{G}$, denoted $\mathcal{E}_\mathcal{G}\in\mathbb{R}^{|\mathbb{E}|\times|\mathbb{V}|}$, is defined such that for edge $e=(i,j)\in \mathbb{E}$, $[\mathcal{E_G}]_{ie} =+1$, $[\mathcal{E_G}]_{je} =-1$, and $[\mathcal{E_G}]_{\ell e} =0$ for $\ell \neq i,j$. 
{We also use simple notions from graph theory \cite{GraphTheory}. A path is a sequence of distinct nodes $v_1,v_2,\cdots,v_n$ such that $\{v_i,v_{i+1}\}\in \mathbb{E}$ for all $i$. A cycle is the union of a path $v_1,\cdots,v_n$ with the edge $\{v_1,v_n\}$. A cycle $v_1,v_2,\cdots,v_{n-1},v_1$ is called simple if $v_i\neq v_j$ for all $i\neq j$. A collection of paths is called vertex-disjoint if no two share a node, except possibly for their first and last nodes.}
Furthermore, for a linear map $T:U\to V$ between vector spaces, we denote the kernel of $T$ by $\ker{T}$.

\section{Network Optimization and MEIP\\ Multi-Agent Systems}\label{sec.background}

The role of network optimization theory in cooperative control was introduced in \cite{SISO_Paper}, and was used in \cite{LCSS_Paper,TAC_Paper} to solve the synthesis problem for MAS.  In this section, we brief on the main results we need from \cite{SISO_Paper, LCSS_Paper, TAC_Paper}.
\vspace{-10pt}
\subsection{Diffusively-Coupled Networks and Their Steady-States}
Consider a collection of {SISO} agents interacting over a network $\mathcal{G}=(\mathbb{V},\mathbb{E})$. The agents $\{\Sigma_i\}_{i\in\V}$ and the controllers $\{\Pi_e\}_{e\in\EE}$ have the following models:
\begin{align} \label{AgentsandControllers}
\Sigma_i: 
\begin{cases}
\dot{x}_i = f_i(x_i,u_i) \\
y_i = h_i(x_i).
\end{cases}, 
\Pi_e: 
\begin{cases}
\dot{\eta}_e = \phi_e(\eta_e,\zeta_e) \\
\mu_e = \psi_e(\eta_e,\zeta_e).
\end{cases}
\end{align} 

\begin{figure} [!t] 
    \centering
    \includegraphics[scale=0.2]{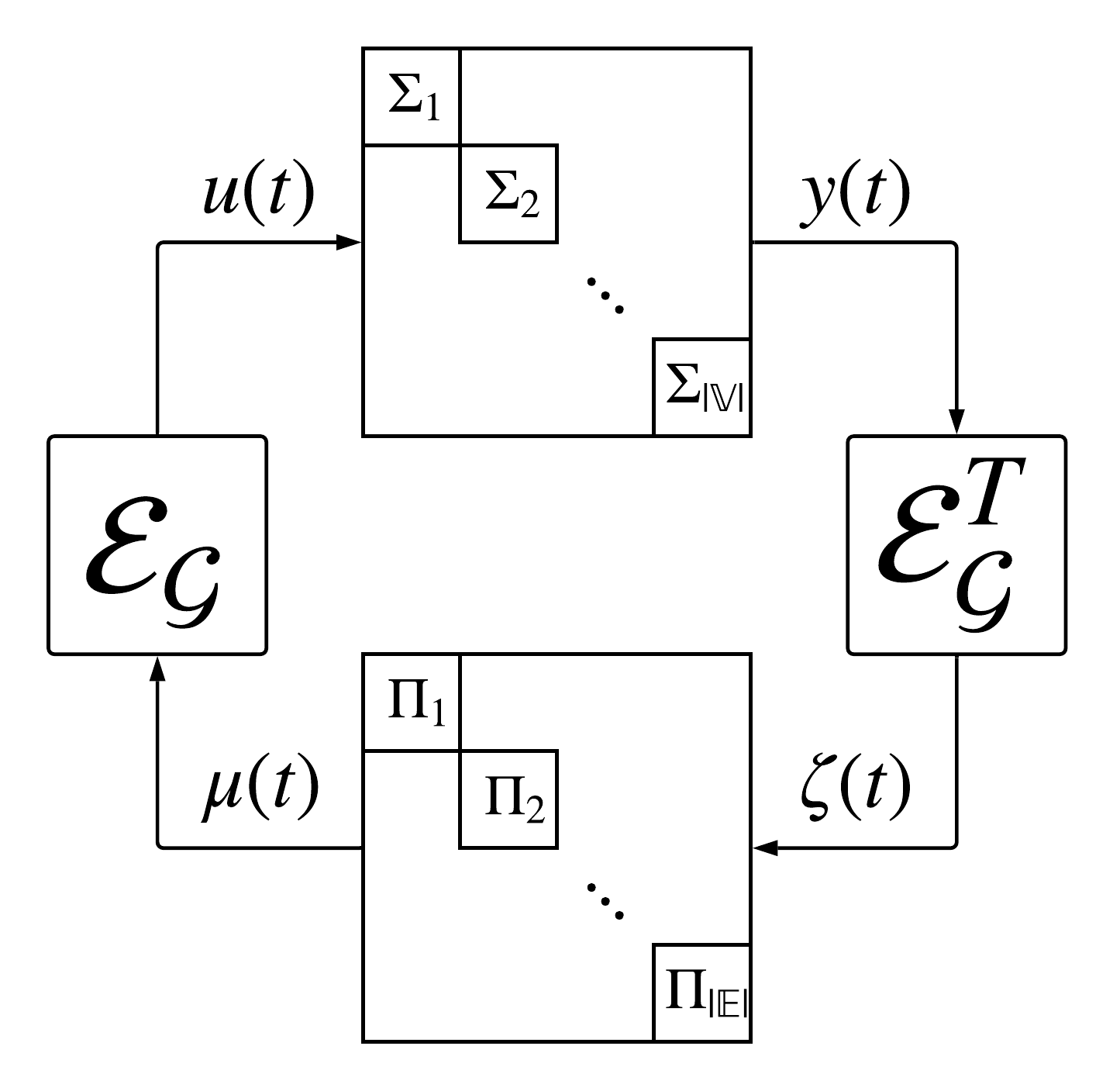}
    \caption{Block-diagram of the closed loop.}
    \label{ClosedLoop}
\vspace{-15pt}
\end{figure}

 We consider stacked vectors of the form $u=[u_1^\top,...,u_{|\mathbb{V}|}^\top]^\top$ and similarly for $x,y,\zeta,\eta$ and $\mu$. 
The agents and controllers are coupled by defining the controller input as $\zeta = \E_\G^\top y$ and the control input as $u = -\E_\G\mu$. This closed-loop system is called a \emph{diffusively-coupled network}, and is denoted by $(\G,\Sigma,\Pi)$. Its structure is illustrated in Figure \ref{ClosedLoop}. We wish to study the steady-states of the closed-loop. Suppose that $(\mathrm{u},\mathrm{y},\upzeta,\upmu)$ is a steady-state of $(\G,\Sigma,\Pi)$. For every $i\in\V,e\in\EE$, $(\mathrm{u}_i,\mathrm{y}_i)$ is a steady-state input-output pair of the $i$-th agent, and $(\upzeta_e,\upmu_e)$ is a steady-state pair of the $e$-th controller. This motivates the following definition, originally from \cite{SISO_Paper}:
\begin{defn} \label{def.SS}
The steady-state input-output relation $k$ of a dynamical system is the collection of all steady-state input-output pairs of the system. Given a steady-state input $\mathrm{u}$ and a steady-state output $\mathrm{y}$, we define
\begin{align}
k(\mathrm u) &= \{\mathrm y:\ (\mathrm{u,y})\in k)\}, ~~ k^{-1}(\mathrm y) &= \{\mathrm u:\ (\mathrm{u,y})\in k)\} .\nonumber
\end{align}
\end{defn}
Let $k_i$ be the steady-state input-output relation for the $i$-th agent, $\gamma_e$ be the steady-state input-output relation for the $e$-th controller, and $k,\gamma$ be their stacked versions. Then the closed-loop steady-state $(\mathrm{u},\mathrm{y},\upzeta,\upmu)$ has to satisfy $\mathrm y\in k(\mathrm u), \upzeta=\E_\G^\top\mathrm y, \upmu\in\gamma(\upzeta), \mathrm u=-\E_\G\upmu$. By a simple manipulation, one can show that $\mathrm{y}$ is a closed-loop steady-state for the agent output if and only if $0\in k^{-1}(\mathrm y) + \E_\G\gamma(\E_\G^\top\mathrm y)$ \cite{TAC_Paper}.
\vspace{-10pt}
\subsection{MEIP Systems and Closed-Loop Convergence}
The convergence of the diffusively-coupled network $(\G,\Sigma,\Pi)$ can be assured using passivity. We first recall the classic definition of (shifted) passivity:
\begin{defn}[\small{Passivity} \cite{Khalil}]
Let $\Upsilon$ be a SISO system with input $u(t)$, output $y(t)$ and state $x(t)$, and let $(\mathrm{u,y})$ be a steady-state input-output pair of the system. For a differentiable function $S=S(x)$ and a number $\rho>0$, we consider the inequality
$\frac{d}{dt}{S(x)} \le -\rho\| y(t) -\mathrm y\|^2 + (y(t)-\mathrm y)(u(t) - \mathrm u)$
We say $\Upsilon$ is passive (w.r.t. $(\mathrm{u,y})$) if there exists a semi-definite storage function $S(x)$ and $\rho \ge 0$ such that the inequality holds for any trajectory. Also, we say the system is output-strictly passive (w.r.t. $(\mathrm{u,y})$) if the same condition holds for some $\rho >0$. The largest number $\rho$ for which the condition holds is called the \emph{(output) passivity index} w.r.t. $(\mathrm{u,y})$.
\end{defn}
{Passivity was first used for diffusively-coupled networkes in \cite{Arcak}. It is known that if $(\mathrm u, \mathrm y,\upzeta,\upmu)$ is an equilibrium of the network, and the agents and controllers are passive with respect to $(\mathrm u_i,\mathrm y_i)$ and $(\upzeta_e,\upmu_e)$, then the network converges to said equilibrium. The existence of an equilibrium for the closed-loop network can be proved using network optimization tools under certain monotonicity assumptions on the steady-state input-output relation of the agents and controllers \cite{SISO_Paper,TAC_Paper}, namely under the following variant of passivity.}

\begin{defn}[\small{Maximal Equilibrium Independent Passivity \cite{SISO_Paper}}]
Consider the SISO dynamical system of the form
\begin{align} \label{Upsilon}
\Upsilon : 
\dot{x} = f(x,u);~~
y = h(x,u),
\end{align}
with input-output relation $r$. The system $\Upsilon$ is said to be (output-strictly) MEIP if the following conditions hold:
\begin{enumerate}
\item[i)] The system $\Upsilon$ is (output-strictly) passive with respect to any steady-state pair $(\mathrm{u},\mathrm{y})$, i.e., {with respect to any $\mathrm u,\mathrm y$ such that} $\mathrm{y} \in r(\mathrm{u})$.
\item[ii)] The steady-state input-output relation $r$ is maximally monotone. That is, if $(\mathrm{u}_1,\mathrm{y}_1),(\mathrm{u}_1,\mathrm{y}_2)\in r$ then $(\mathrm{u}_1-\mathrm{u}_2)(\mathrm{y}_1-\mathrm{y}_2) \ge 0$, and $r$ is not contained in any larger monotone relation \cite{Rockafellar1998}.
\end{enumerate}
The passivity index of the system $\Upsilon$ is defined as $\underset{\mathrm y\in\mathrm r(\mathrm u)}{\min} \rho_{\mathrm{u,y}}$, where $\rho_{\mathrm{u,y}}$ is the passivity index with respect to $(\mathrm{u,y})$.
\end{defn}
Such systems include single integrators, gradient systems, port-Hamiltonian systems on graphs, and others (see \cite{SISO_Paper, TAC_Paper} for more examples). In this work we often consider networks of control-affine systems. The theorem below gives a sufficient condition for a control-affine system to be MEIP:
\begin{thm} \label{thm.MEIPAgents}
Let $\Sigma$ be the SISO system of the form $\dot{x} = -f(x)+q(x)u, y=h(x)$. Suppose that $q(x)$ is positive for all $x$, that $h$ is strictly monotone $C^1$ ascending, and that $f/q$ is $C^1$ and monotone ascending.
\begin{enumerate}
\item A pair $(\mathrm{u,y})$ is a steady-state input-output pair for $\Sigma$ if and only if there exists some $\mathrm x \in \R$ such that $\mathrm u = f(\mathrm x)/q(\mathrm x)$ and $\mathrm y = h(\mathrm x)$;
\item For any $\mathrm x\in \mathbb{R}$, the function 
$S(x) = \int_{\mathrm x} ^ x \frac{h(\sigma)-h(\mathrm x)}{q(\sigma)}d\sigma$ is a storage function for the steady-state input-output pair $\mathrm u = f(\mathrm x)/q(\mathrm x)$ and $\mathrm y = h(\mathrm x)$;
\item The function $S(x)$ proves that $\Sigma$ is passive w.r.t. $(\mathrm u,\mathrm y)$ with passivity index 
$\rho = \inf_{x\in \R} \frac{\frac{f(x)}{q(x)}-\frac{f(\mathrm x)}{q(\mathrm x)}}{h(x)-h(\mathrm x)}\ge 0;$
\item If either $\lim_{|t|\to\infty} |f(t)/q(t)| = \infty$ or $\lim_{|t|\to\infty} |h(t)| = \infty$, then the system is MEIP.
\item If the derivative of $h$ is always positive, then the inverse steady-state relation $k^{-1}$ is differentiable.
\end{enumerate}
\end{thm}
\begin{proof}
The first, second and fourth parts are proved in \cite[Proposition~1]{Sharf2019e}. As for the third, we note that:
\begin{align*}
\dot{S} =& \frac{(h(x)-h(\mathrm x))}{q(x)}\dot{x} = \frac{(h(x)-h(\mathrm x))}{q(x)}(-f(x) + q(x)u) = \\
=& (h(x) - \mathrm y)(u-\mathrm u) -(h(x)-h(\mathrm x))\bigg(\frac{f(x)}{q(x)}-\frac{f(\mathrm x)}{q(\mathrm x)}\bigg) \\
\le&(y-\mathrm y)(u-\mathrm u)-\rho(y-\mathrm y)^2 
\end{align*}
We note that $\rho \ge 0$ as $(h(x)-h(\mathrm x))\bigg(\frac{f(x)}{q(x)}-\frac{f(\mathrm x)}{q(\mathrm x)}\bigg)\ge 0$ by monotonicity. Moreover, we note that $\int_{\mathrm x}^{x} \frac{h(\sigma)-h(\mathrm x)}{q(\sigma)}d\sigma \ge 0$, with strict inequality whenever $x\neq \mathrm x$, as $h$ is strictly monotone and $q(x) > 0$. Thus $S$ is a $C^1$ storage function, and we conclude the system is passive with passivity index $\rho\ge 0$ with respect to the steady-state input-output pair $(\mathrm{u,y})$. 
Lastly, we note that because $h$ is strictly monotone, the inverse $h^{-1}$ can be defined. Thus, the inverse steady-state relation $k^{-1}$ is given by $k^{-1}(\mathrm y) = \frac{f(h^{-1}(\mathrm y))}{q(h^{-1}(\mathrm y))}$, which is differentiable by the inverse function theorem.
\end{proof}

As we previously remarked, MEIP can be used  to prove existence of a closed-loop equilibrium for networks:

\begin{thm}[\cite{SISO_Paper,LCSS_Paper}]\label{thm.ClosedLoopSteadyStates}
Consider the network $(\G,\Sigma,\Pi)$. Assume the agents $\Sigma_i$ are MEIP, and the {controllers} $\Pi_e$ are output-strictly MEIP (or vice versa). Then the signals $u(t),y(t),\zeta(t),\mu(t)$ of the closed-loop system converge to steady-state values $\mathrm{u},\mathrm{y},{\upzeta},{\upmu}$, where $0\in k^{-1}(\mathrm{y}) + \E_\G \gamma(\E_\G^\top \mathrm{y})$.
\end{thm}

\if(0)
{The theorem can also be restated in the language of optimization. Indeed, because the plants and controllers are MEIP, a theorem by Rockafellar implies that there exists convex functions $K,\Gamma$ such that the subdifferentials $\partial K,\partial K^\star, \partial \Gamma, \partial \Gamma^\star$ are equal to $k,k^{-1},\gamma,\gamma^{-1}$ respectively, where $^\star$ denotes the Legendre Transform. Thus, $\mathrm y$ satisfies the steady-state ``equation" $0\in k^{-1}(\mathrm{y}) + \E_\G \gamma(\E_\G^\top \mathrm{y})$ if and only if $(\mathrm y,\E_\G^\top \mathrm y)$ is a minimizer of the optimal potential problem (OPP) seen below. In that case, one can show that the limit of $(\mu(t),u(t))$, as guaranteed by the theorem, converges the the minimizer of the dual problem, known as the optimal flow problem (OFP):
\begin{center}
\begin{tabular}{ c||c }
 \textbf{Optimal Potential Problem}  & \textbf{Optimal Flow Problem}   \\
 (OPP) & (OFP)  \\\hline
 $ \begin{array}{cl} \underset{\mathrm{y,\zeta}}{\min} &K^\star(\mathrm{y}) + \Gamma(\zeta)\\
s.t.&\E_\G^\top\mathrm{y} = \zeta 
\end{array} $&  $ \begin{array}{cl}\underset{\mathrm{u,\mu}}{\min}& K(\mathrm{u}) + \Gamma^\star(\mu) \\
s.t. &\mathrm{u} = -{\E_\G}\mu.
\end{array} $ 
\end{tabular}
\end{center}
These two dual problems are some of the main problems studied in the field of network optimization, which studies optimization problems defined on networks, and have been extensively studied over the last few decades. See \cite{SISO_Paper, TAC_Paper} for a more expansive discussion about these problems and the network optimization framework in general.}
\fi
\vspace{-15pt}
\subsection{The Synthesis Problem for MEIP Multi-Agent Systems}
The synthesis problem of MAS with MEIP agents has been studied in \cite{LCSS_Paper,TAC_Paper}. The problem deals with synthesizing controllers $\{\Pi_e\}$ forcing the closed-loop network to converge to some desired steady-state output $\mathrm y^\star$, when the agents $\Sigma$ and the graph $\G$ are known. We cite the following results from \cite{TAC_Paper}:
\begin{thm}[\cite{TAC_Paper}] \label{thm.ClosedLoopSynthesis}
Let $\Sigma$ be any MEIP agents and let $\G$ be any graph. {Let $\mathrm y^\star \in \R^{|\V|}$ be any desired steady-state output. Then there exists a solution to the synthesis problem (i.e., a realization of the controllers $\Pi$) with desired output $\mathrm y^\star$ for which the controllers are output-strictly MEIP.}
\end{thm}

{
\begin{rem} \label{rem.FormReconfig}
The paper \cite{TAC_Paper} depicts many possible solutions to the synthesis problem with output-strictly MEIP controllers. It is shown that one can always solve the problem using affine controllers. Another suggested solution is an augmentation of any preferred output-strictly MEIP controller with a constant exogenous input. In practice, we will usually opt for the augmentation procedure when using the theorem as a tool for synthesis, as many real-world {networks} are already equipped with some {given edge controllers}. If this is not the case, one can use the affine controllers instead.
\end{rem}
}
\vspace{-15pt}
\section{Problem Formulation}\label{sec.probform}
This section presents the problem we aim to solve, and states the assumptions we make to tackle it. We consider a {diffusively-coupled network} of the form $\N_\G=(\G,\{\Sigma_i\}_{i\in\V},\{\Pi_e\}_{e\in\EE_\G})$, where $\G = (\V,\EE_\G)$ is the interaction graph, $\Sigma_i$ are the agents, and $\Pi_e$ are the edge controllers. For any subgraph $\HH=(\V,\EE_\HH)$ of $\G$, we can consider another diffusively-coupled network $\N_\HH=(\HH,\{\Sigma_i\}_{i\in\V},\{\Pi_e\}_{e\in\EE_\HH})$. We can think of $\N_\HH$ as a faulty version of $\N_\G$, in which the edge controllers corresponding to the edges $\EE_\G\setminus \EE_\HH$ have malfunctioned and stopped working. Edges can fault mid-run, but we assume that once an edge has malfunctioned, it remains faulty for the remainder to the run. If we let $\mathfrak{G}$ be the collection of all nonempty subgraphs of $\G$, then one can think of the closed-loop diffusively-coupled network as a switched system, where the switching signal $\varsigma:[0,\infty)\to\mathfrak{G}$ designates the functioning edges at each time instant. The assumption that faulty edges remain faulty throughout the run can be described using the switching signal $\varsigma$. Namely, we require that the switching signal $\varsigma$ is \emph{non-increasing}, in the sense that for all times $t_1 < t_2$, $\varsigma(t_2)$ is a subgraph of $\varsigma(t_1)$. We denote the number of faulty edges at time $t$ by $\hat{\varsigma}(t)$.

Now, consider a collection of agents $\{\Sigma_i\}$ and a graph $\G$. Fix some constant vector $\mathrm y^\star \in \R^{|\V|}$. Our goal is to design a control scheme for which the closed-loop network will converge to the steady-state output $\mathrm y^\star$. In the absence of faults, we can solve the synthesis problem as in Theorem \ref{thm.ClosedLoopSynthesis}. However, designing controllers while ignoring faults might prevent the system from achieving the control goal. {For that reason, we also seek for a \emph{fault monitoring system}, consisting of the agents and networked controllers, that attempts to identify aberrant behavior. When it does, it declares a fault \cite{Willsky1976}.} We now formulate the problems of {network} fault detection and isolation:
\begin{prob}[Network Fault Detection] \label{prob.NFD}
Let $\{\Sigma_i\}_{i\in\V}$ be a set of agents, $\G$ be a graph, $\mathrm y^\star$ be any desired steady-state output, and let $\varsigma(t)$ be any non-increasing switching signal. Find edge controllers $\{\Pi_e\}_{e\in \mathbb{E}_\G}$ and a {fault monitoring system} such that,
\begin{enumerate}
\item[i)] if no faults occur, i.e. $\varsigma(t) = \G,\ \forall t$, then the closed-loop {diffusively-coupled network} converges to the steady-state output $\mathrm y^\star$, {and the fault monitoring system never declares a fault};
\item[ii)] if faults do occur, i.e. $\exists t, \varsigma(t) \neq \G$, then {the fault monitoring system declares a fault}.
\end{enumerate}
\end{prob}
\begin{prob}[Network Fault {Detection and} Isolation] \label{prob.NFI}
Let $\{\Sigma_i\}_{i\in\V}$ be a set of agents, $\G$ be a graph, and $\mathrm y^\star$ be any desired steady-state output. 
Given some $r<|\EE_\G|$, find a synthesis for the edge controllers such that for any monotone non-increasing switching signal $\varsigma$ such that $\hat{\varsigma}(t)\le r,\ \forall t$, the closed-loop {diffusively-coupled network} converges to the steady-state output $\mathrm y^\star$, {i.e., the effect of up to $r$ faults can be isolated from the network.\footnote{Some authors refer to fault isolation in this case as identifying the faulty links, which is achieved by the algorithms in Subsection \ref{subsec.FI} as a side effect.}}
\end{prob}

\vspace{-15pt}
\subsection{Assumptions} \label{subsec.Assumptions}
We now state the assumptions used throughout the work. For the remainder of this work, we fix the agents $\{\Sigma_i\}$, and make the following assumption.
\begin{assump} \label{Assumption}
The agent dynamics $\{\Sigma_i\}$ are MEIP, and the chosen controller dynamics $\{\Pi_e\}$ are output-strictly MEIP (or vice versa). Moreover, the relations $k^{-1}_i$, $\gamma_e$ are $C^1$ functions. Furthermore, the derivative $\frac{dk_i^{-1}}{d\mathrm y_i}$ is positive at any $\mathrm y_i \in \mathbb{R}$.
\end{assump}
The passivity assumption assures that all the systems $\N_\HH$ will globally asymptotically converge to some limit. The added smoothness assumptions, together with the positive derivative assumption, are technical assumptions that are needed to apply tools from manifold theory that will be used later. {The passivity assumption allows the consideration of, e.g. port-Hamiltonian systems and gradient-descent systems \cite{SISO_Paper}. Moreover, if a system satisfies any dissipation inequality with respect to all equilibria, one can use output-feedback and input-feedthrough to force MEIP \cite{Sharf2019a}. Theorem \ref{thm.MEIPAgents} shows that the smoothness assumption holds for many control-affine systems. Moreover, it can be easily shown using the definition of passivity that if $\Sigma_i$ is output-strictly MEIP with passivity index $\rho$, then $\frac{dk_i^{-1}}{d\mathrm y_i} \ge \rho>0$ whenever $k_i^{-1}$ is differentiable. 
Furthermore, the smoothness assumption can be relaxed by allowing $k^{-1}_i,\gamma_e$ to not be differentiable at finitely many points. The arguments presented below still hold, but require heavier tools from measure theory, so we avoid them for clarity of presentation.}

In some cases, we need to sense the state of the system, including the state of the controllers. Sometimes, the control model is such that the controller state has a physical meaning that can be measured even for non-connected agents. For example, in the traffic control model in \cite{Bando1995}, the state $\eta_e$ is the relative position between two vehicles. However, the controller state of some systems might not have a physical meaning. For example, consider a collection of robots trying to synchronize their positions, where the output $y(t)$ is the position of each robot and the edge controllers are PI controllers. In that case, the controller state $\eta(t)$ has no physical meaning, and thus cannot be defined for non-connected agents. Some of the techniques developed later require us to be able to sense the state of the system, including the controllers' states. Thus, we will sometimes make the following assumption: 
\begin{assump} \label{Assumption_g_1}
The controllers $\Pi_e$ are static nonlinearities given by the functions $g_e$, i.e., $\mu_e{(t)} = g_e(\zeta_e{(t)})$ {for all $t$}. In this case, the steady-state relation $\gamma_e$ is equal to the function $g_e$, and the closed-loop system is $\dot{x} = f(x,-\E_\G g(\E_\G^\top h(x)))$, or equivalently, $\dot{x}_i = f_i\left(x_i,\sum_{e=\{i,j\}} g_e(h_j(x_j)-h_i(x_i))\right)$.
\end{assump}
In one of the methods below, we will want to have a clear relationship between the measurements $h_i(x_i)$ and the storage functions $S_i(x_i)$. To achieve this, we follow Theorem \ref{thm.MEIPAgents} and assume that the agents are control-affine:
\begin{assump} \label{Assumption_g_2}
Assumption \ref{Assumption_g_1} holds, and the agents have the form $\dot{x}_i = -f_i (x_i) + q_i(x_i) u_i ;\ y_i = h_i(x_i)$. Thus, the closed-loop system is governed by:
\begin{align}
\dot{x}_i = -f_i(x_i) + q_i(x_i) \sum_{e=\{i,j\}} g_e(h_j(x_j) - h_i(x_i)).
\end{align}
\end{assump}

It should be noted that the MEIP property for the static controllers $g_e$ reduces to monotonicity of the functions $g_e$. 

In the next section, we start heading toward a solution to Problems \ref{prob.NFD} and \ref{prob.NFI}. We do so by exhibiting a method for asymptotically differentiating between the nominal dynamical system $\N_\G$ and the faulty dynamical systems $\N_\HH$. Later, we show how this asymptotic differentiation can induce a finite-time differentiation of the systems.

\section{Asymptotic Differentiation\\ Between Networks} \label{sec.IndicationVectors}
In this section, we develop the notion of edge-indication vectors, which will be used for {network} fault detection later. In \cite{Sharf2018b}, the notion of indication vectors was first developed. These are constant exogenous inputs used to drive the closed-loop system, chosen appropriately to give different steady-state limits for systems with identical agents and controllers, but different underlying graphs. The idea of using constant exogenous inputs to drive the system into favorable steady-state outputs was also used in \cite{Sharf2019d} to give a network reconstruction algorithm with optimal time complexity, although it considers sets of multiple constant exogenous inputs applied in succession. Here, we opt for a slightly different strategy. 

In \cite{Sharf2018b,Sharf2019d}, the problem of network reconstruction was considered, in which we cannot affect the agents, controllers, or the underlying graph. In network FDI, we are doing synthesis, so we can manipulate the controllers and (in most cases) the underlying network. For that reason, we opt for a slightly different idea, in which we add a constant exogenous signal to the \emph{output of the controllers}, that is, we consider $u(t) = -\E_\G(\mu(t) + \mathrm w)$. A system implementing this control law is said to have the interaction protocol $(\Pi,\mathrm w)$. Analogously to the notion of indication vectors, we desire that networks with identical agents and controllers, but different underlying graphs, will be forced to converge to different steady-state outputs. This is because we can monitor the output $y$ of the system and use it to detect changes in the underlying graph, i.e., network faults. For that, we first determine what the steady-state limit is for these systems $(\G, \Sigma, (\Pi,\mathrm{w}))$.

\begin{prop} \label{prop.Indication}
Consider a diffusively-coupled network $\N_\HH = (\HH,\Sigma,\Pi)$ satisfying Assumption \ref{Assumption}. Suppose that $\mathrm w\in\mathbb{R}^{|\EE_\HH|}$ is any constant signal added to the controller output, i.e., the loop is closed as $u(t) = -\E_\HH(\mu(t) + \mathrm w)$.
Then $\mathrm y$ is a steady-state output of the closed-loop system if and only if
\begin{align}\label{eq.SteadyState_1}
k^{-1}(\mathrm y) + \E_\HH \gamma (\E_\HH^\top \mathrm y) = -\E_\HH \mathrm w.
\end{align}
\end{prop}
\begin{proof}
Follows from the {discussion after Definition \ref{def.SS}}, as the new steady-state relation for the controllers is given as $\tilde{\gamma}(\zeta) = \gamma(\zeta) + \mathrm w$.
\end{proof}

In our case, the constant signal $\mathrm w$ will be in $\R^{|\EE_\G|}$, as we determine the exogenous controller output on each edge of $\G$. If one then considers the system $\N_\HH$ for some $\HH\in\mathfrak{G}$, then the exogenous controller output will be different from $\mathrm w$, as it will only have entries of $\mathrm w$ corresponding to edges in $\HH$. To formulate this, take any graph $\HH\in\mathfrak{G}$, and let $P_\HH$ be the linear map $\R^{|\EE_\G|}\to\R^{|\EE_\HH|}$ removing entries corresponding to edges absent from $\HH$. In other words, this is a $\mathbb{R}^{|\EE_\HH|\times |\EE_\G|}$ matrix with entries in $\{0,1\}$, whose rows are the rows of the identity matrix $\mathrm{Id}\in \R^{|\EE_\G|\times |\EE_\G|}$ corresponding to the edges of $\HH$.

We can now define the notion of edge-indication vectors.
\begin{defn}\label{def.IndicationVectors}
Let $(\G,\Sigma,\Pi)$ be a network satisfying Assumption \ref{Assumption}. Let $\mathrm w \in \R^{|\EE_\G|}$ by any vector, and for any graph $\HH \in \mathfrak{G}$, we denote the solution of \eqref{eq.SteadyState_1} with underlying graph $\HH$ and exogenous input $P_\HH \mathrm w$ by $\mathrm y_\HH$. 
\begin{itemize}
\item The vector $\mathrm w$ is called a $(\mathfrak{G},\HH)$-edge-indication vector if for any $\HH^\prime \in \mathfrak{G}$ {such that $\HH^\prime \neq \HH$, we have} $\mathrm y_{\HH} \neq \mathrm y_{\HH^\prime}$.
\item The vector $\mathrm w$ is called a $\mathfrak{G}$-edge-indication vector if for any two graphs $\HH_1\neq \HH_2$ in $\mathfrak{G}$,  $\mathrm y_{\HH_1} \neq \mathrm y_{\HH_2}$.
\end{itemize}
\end{defn}

\begin{note}
An edge-indication vector is a bias chosen on each edge in $\G$. This bias can be programmed into the controllers and nodes, and need not be changed nor computed on-line. In this light, for any $\mathrm w \in \R^{|\EE_\G|}$, \eqref{eq.SteadyState_1} transforms into
\begin{align}\label{eq.SteadyState}
k^{-1}(\mathrm y) + \E_\HH \gamma (\E_\HH^\top  \mathrm y) = -\E_\HH P_\HH \mathrm w,
\end{align}
\end{note}

We wish to find a $\mathfrak G$-edge-indication vector for given agents and controllers, or at least a $(\mathfrak G,\G)$-edge-indication vector. As in \cite{Sharf2018b}, we use randomization. We claim that random vectors are $\mathfrak G$-edge-indication vectors with probability $1$.
\begin{thm} \label{thm.GeneralRandomIndication}
Let $\PP$ be any absolutely continuous\footnote{{Unless stated otherwise, absolute continuity is with respect to the Lebesgue measure.}} probability measure on $\R^{|\EE_\G|}$. Let $\mathrm w$ be a vector sampled according to $\PP$. Then $\PP(\mathrm w \text{ is a $\mathfrak{G}$-edge-indication vector}) = 1$.
\end{thm}
\begin{proof}
From the definition, $\mathrm w$ is not a $\mathfrak G$-edge-indication vector if and only if there are two graphs $\G_1,\G_2\in\mathfrak G$ such that the same vector $\mathrm y$ solves equation \eqref{eq.SteadyState} for both graphs. We show that for any $\G_1,\G_2 \in \mathfrak G$, the probability that the two equations share a solution is zero. 

Let $n$ be the number of vertices in $\G$. For each graph $\HH \in \mathfrak{G}$, define a function $F_\HH:\mathbb{R}^n\times \mathbb{R}^{|\EE_\G|} \to \mathbb{R}^n$ by $F_\HH(\mathrm y,\mathrm w) = k^{-1}(\mathrm y) + \E_\HH \gamma(\E_\HH^\top  \mathrm y) + \E_\HH P_\HH \mathrm w$. The set of steady-state exogenous input and output pairs for the system $\N_\HH$ is given by the set $\mathcal A_\HH = \{(\mathrm{y,w}):\ F_\HH(\mathrm{y,w}) = 0\}$. 
We note that the differential $dF_\HH$ always has rank $n$. Indeed, it can be written as $[\nabla k^{-1}(\mathrm y) + \E_\HH \nabla \gamma(\E_\HH^\top  \mathrm y), \E_\HH P_\HH]$, {where $\nabla \gamma(\E_\HH^\top \mathrm y)\in \R^{|\E_\HH|\times n}$}. {By assumption \ref{Assumption}, the first matrix, of size $n\times n$, is positive-definite as a sum of a positive-definite matrix and a positive semi-definite matrix}, hence invertible. Thus, by the implicit function theorem, $\mathcal{A}_\HH$ is a manifold of dimension $|\EE_\G|$.
Moreover, by Assumption \ref{Assumption}, {for any $\mathrm w$ there is a unique $\mathrm y$ such that \eqref{eq.SteadyState_1} is satisfied. Thus, $\PP$ gives rise to an absolutely continuous\footnote{{With respect to the $|\EE_\G|$-dimensional Hausdorff measure, or equivalently, with respect to the standard Riemannian volume form on $\mathcal{A}_\mathcal{H}$.}} probability measure on each manifold $\mathcal A_\HH$.\footnote{{As the push-forward measure of $\mathbb{P}$ under the map $\mathrm w\mapsto (\varphi(\mathrm w),\mathrm w)$, where $\varphi$ is the local map $\mathrm w \mapsto \mathrm y$ given by the implicit function theorem.}}} Hence, it is enough to show that for any $\G_1\neq \G_2$, the intersection $\mathcal A_{\G_1}\cap \mathcal A_{\G_2}$ has dimension $\le |\EE_\G|-1$.

To show this, we take any point $(\mathrm{y,w}) \in \mathcal A_{\G_1}\cap \mathcal A_{\G_2}$. As both $\mathcal A_{\G_1}, \mathcal A_{\G_2}$ are of dimension $|\EE_\G|$, it is enough to show that they do not have the same tangent space at $(\mathrm{y,w})$. The tangent space of the manifold $\mathcal A_\HH$ is given by the kernel of the differential $dF_{\HH}(\mathrm{y,w}) : \mathbb{R}^n\times \mathbb{R}^{|\EE_\G|} \to \mathbb{R}^n$, so we show that if $\G_1\neq \G_2$, the kernels $\ker dF_{\G_1} , \ker dF_{\G_2}$ are different at $(\mathrm{y,w})$. As $\G_1\neq\G_2$, we can find an edge existing in one of the graphs and not the other. Assume without loss of generality that the edge $e$ exists in $\G_1$ but not in $\G_2$, and let $v=(0,\mathbbm{1}_e)$, where $\mathbbm{1}_e$ is the vector in $\mathbb{R}^{|\EE_\G|}$ with all entries zero, except for the $e$-th entry, which is equal to $1$. Then $v \in \ker dF_{\HH}$ if and only if $\mathbbm{1}_e\in \ker (\E_\HH P_\HH)$. It is clear that $\mathbbm{1}_e\not \in \ker (\E_{\G_1}P_{\G_1})$, {as $P_{\G_1}\mathbbm{1}_e = \mathbbm{1}_e$,} and thus $\E_{\G_1}P_{\G_1}\mathbbm{1}_e = \E_{\G_1}\mathbbm{1}_e \neq 0$. Moreover, $\mathbbm{1}_e\in \ker (\E_{\G_2}P_{\G_2})$, as $P_{\G_2}\mathbbm{1}_e = 0$, so $\ker dF_{\G_1} \neq \ker dF_{\G_2}$ at $(\mathrm{y,w})$. Thus $\mathcal{A}_{\G_1} \cap \mathcal{A}_{\G_2}$ is of dimension $\le |\EE_\G|-1$, meaning that it is a zero-measure set inside both $\mathcal{A}_{\G_1},\mathcal{A}_{\G_2}$.
\end{proof}

Theorem \ref{thm.GeneralRandomIndication} presents a way to choose a $\mathfrak{G}$-edge-indication vector, but does not deal {with} the control goal. One could satisfy the control goal by using Theorem \ref{thm.ClosedLoopSynthesis} to solve the synthesis problem for the original graph $\G$, but we cannot assure we get an edge-indication vector. 
Note that any $\mathrm w \in \ker \E_\G P_\G$ gives a solution of \eqref{eq.SteadyState} identical to the solution for $\mathrm w = 0$. Thus, choosing an exogenous control input in $\ker \E_\G P_\G$ does not change the steady-state output of the system $\N_\G$. However, it does change the steady-state output of all other systems $\N_\HH$. This suggests to search for an edge-indication vector in $\ker \E_\G P_\G$. We show that this is possible if $\G$ is ``sufficiently connected", defined below in an exact manner.

\begin{prop}[Menger's Theorem \cite{GraphTheory}]
Let $\G$ be any connected graph. The following conditions are equivalent:
\begin{enumerate}
\item Between every two nodes there are $k$ vertex-disjoint simple paths.
\item For any $k-1$ vertices $v_1,\cdots,v_{k-1}\in \V$, the graph $\G-\{v_1,\cdots,v_{k-1}\}$ is connected.
\end{enumerate}
Graphs satisfying either of these conditions are called $k$-connected graphs. 
\end{prop}

We will take special interest in $2$-connected graphs. Specifically, we can state the following theorem about edge-indication vectors in $\ker \E_\G P_\G$.

\begin{thm}\label{thm.SynthesisRandomIndication}
Let $\PP$ be any absolutely continuous probability distribution on $\ker \E_\HH P_\HH$, where $\HH$ is 2-connected. Suppose furthermore that $\mathrm w$ is a vector sampled according to $\PP$. Then $\PP(\mathrm w \text{ is a ($\mathfrak{G},\HH)$-edge-indication vector}) = 1$.
\end{thm}

We first need to state and prove a lemma:
\begin{lem}
Let $\HH$ be a 2-connected graph. Suppose we color the edges of $\HH$ in two colors, red and blue. If not all edges have the same color, then there is a simple cycle in $\HH$ with both red and blue edges.
\end{lem}

\begin{proof}
Suppose, heading toward contradiction, that any simple cycle in $\HH$ is monochromatic. We claim that for each vertex $x$, all the edges touching $x$ have the same color. Indeed, take any vertex $x$, and suppose that there are two neighbors $v_1,v_2$ of $x$ such that the edge $\{x,v_1\}$ is blue and the edge $\{x,v_2\}$ is red. We note that $v_1\to x\to v_2$ is a path from $v_1$ to $v_2$, meaning there is another path from $v_1$ to $v_2$ which does not pass through $x$ . Adding both edges to the path yields a simple cycle with edges of both colors, as $\{x,v_1\}$ is blue and $\{x,v_2\}$ is red. Thus, every node touches edges of a single color.

Let $\V_\text{red}$ be the set of nodes touching red edges, and $\V_\text{blue}$ be the set of nodes touching blue edges. We know that $\V_\text{red}$ and $\V_\text{blue}$ do not intersect. Moreover, if we had an edge between $\V_\text{red}$ and $\V_\text{blue}$, it had a color. Assume, without loss of generality, it is blue. That would mean some vertex in $\V_\text{red}$ would touch a blue edge, which is impossible. Thus there are no edges between $\V_\text{red}$ and $\V_\text{blue}$. By assumption, there is at least one edge of each color in the graph, meaning that both sets are nonempty. Thus we decomposed the set of vertices in $\HH$ to two disjoint, disconnected sets. As $\mathcal{H}$ is a connected graph, we find a contradiction and complete the proof.
\end{proof}

We can now prove Theorem \ref{thm.SynthesisRandomIndication}.
\begin{proof}
We denote $m_1 = \dim \ker \E_\HH P_\HH$. The proof is similar to the proof of Theorem \ref{thm.GeneralRandomIndication}. We again define functions $F_{\G_1}$ for graphs $\G_1 \in \mathfrak G_n$ as $F_{\G_1}(\mathrm{y,w}) = k^{-1}(\mathrm y) + \E_{\G_1} g(\E_{\G_1}^\top  \mathrm y) + \E_{\G_1} P_{\G_1} \mathrm w$, but this time we consider the function $F_{\G_1}$ as defined on the space $\ker \E_\HH P_\HH\subset \mathbb{R}^{|\EE_\G|}$. As before, we define $\mathcal A_{\G_1} = \{(\mathrm{y,w} : F_{\G_1}(\mathrm{y,w}) = 0\}$ and use the implicit function theorem to show that $\mathcal A_{\G_1}$ are all manifolds, but their dimension this time is $m_1 = \dim \ker \E_\HH P_\HH$. This time, we want to show that if $\HH \neq \G_1$, then $\mathcal A_\HH \cap \mathcal A_{\G_1}$ is an embedded sub-manifold of dimension $\le m_1 - 1$, as we want to show that (with probability 1), the solutions \eqref{eq.SteadyState} with graph $\G_1$ and graph $\HH$ are different. As before, it is enough to show that if $(\mathrm{y,w})\in \mathcal A_{\G_1} \cap \mathcal A_\HH$ then the kernels $\ker dF_{\G_1}$ and $\ker dF_\HH$ are different at $(\mathrm{y,w})$. We compute that for any graph ${\G_1}$,
\begin{align*}
dF_{\G_1} = [\nabla k^{-1}(\mathrm y) + \E_{\G_1} \nabla \gamma(\E_{\G_1}^\top  \mathrm y),\ (\E_{\G_1} P_{\G_1})|_{\ker \E_\HH P_\HH}],
\end{align*}
where $\cdot|_{\ker{\E_\HH P_\HH}}$ is the restriction of the matrix to $\ker{\E_\HH P_\HH}$. Thus, if $\G_1$ is any graph in $\mathfrak{G}$ which is not a subgraph of $\HH$, it contains an edge $e$ absent from $\HH$. Following the proof of Theorem \ref{thm.GeneralRandomIndication} word-by-word, noting that $\mathbbm{1}_e \in \ker \E_\HH P_\HH$, we conclude that the $\ker dF_{\G_1}$ and $\ker dF_\HH$ are different at $(\mathrm{y,w})$. Thus we restrict ourselves to non-empty subgraphs $\G_1$ of $\HH$.

For any collection $E$ of edges in $\EE_\HH$, we consider $v=(0,\mathbbm{1}_E)$, where $\mathbbm{1}_E$ is equal to $\sum_{e\in E} \mathbbm{1}_e$. If $E$ is a the set of edges of a cycle in $\HH$, then the vector $v$ lies in the kernel of $dF_{\HH}$. We show that there is some cycle in $\HH$ such that $v$ does not lie in the kernel of $dF_{\G_1}$, completing the proof.

The graph $\G_1$ defines a coloring of the graph $\HH$ - edges in $\G_1$ are colored in blue, whereas edges absent from $\G_1$ are colored in red. Because $\G_1$ is a non-empty proper subgraph of $\HH$, this coloring contains both red and blue edges. By the lemma, there is a simple cycle in $\HH$ having both red and blue edges. Let $E$ be the set of the edges traversed by the cycle. We claim that $\E_{\G_1} P_{\G_1} \mathbbm{1}_E \neq 0$, which will complete the proof of the theorem. Indeed, because the simple cycle contains both red and blue edges, we can find a vertex touching both a red edge in the cycle and a blue edge in the cycle. We let $v$ be the vertex, and let $e_1,e_2$ be the corresponding blue and red edges. Recalling the cycle is simple, these are the only cycle edges touching $v$. However, by the coloring, $e_1$ is in $\G_1$, but $e_2$ is not. Thus,
\begin{align*}
(\E_{\G_1} P_{\G_1} \mathbbm{1}_E)_v &= (\E_{\G_1})_{ve_1} (P_{\G_1})_{e_1e_1} + (\E_{\G_1})_{ve_2} (P_{\G_1})_{e_2e_2} \\&= (\E_{\G_1})_{ve_1} = \pm 1 \neq 0,
\end{align*}
and in particular, $\mathbbm{1}_E \not\in \ker \E_{\G_1} P_{\G_1}$.
\end{proof}

\section{Monitoring Network Faults} \label{sec.Applications}
In this {section}, we consider two applications of the developed framework, namely {network} fault detection and isolation, and defense strategies for adversarial games over networks. We first present a simple algorithm for network fault detection. Then, we discuss defense strategies for adversarial games over networks, which will require a bit more effort. Lastly, we exhibit a network fault isolation protocol, which will be a combination of the previous two algorithms, which will be demonstrated in a case study in Section \ref{sec.Simulation}. In order to apply the framework of edge-indication vectors, we need an algorithm elevating the asymptotic differentiation achieved in the previous section to an on-line differentiation scheme. Thus, we make the following assumption:
\begin{assump}\label{assump.Diff}
There exists an algorithm $\mathscr{A}$ which receives a model for a diffusively-coupled network $(\G,\Sigma,\Pi)$ and a conjectured limit $\mathrm y^\star$ as input, and takes measurements of the network in-run. The algorithm stops and declares ``no" if and only if the network does not converge to $\mathrm y^\star$, and otherwise runs indefinitely. The algorithm $\mathscr{A}$ is called a convergence assertion algorithm.
\end{assump}
{In the language of Section \ref{sec.probform}, this is a fault-monitoring system that never gives false positives nor false negatives.} For now, we assume such algorithm exists. We will discuss this assumption in Section \ref{sec.ConvAssert}.
\subsection{Detecting Network Faults}\label{subsec.FaultDetection}
We first focus on Problem \ref{prob.NFD}. To tackle the problem, we use the notion of edge-indication vectors from Section \ref{sec.IndicationVectors}. Suppose we have MEIP agents $\{\Sigma_i\}$. We first take any output-strictly MEIP controllers $\{\Pi_e\}$ solving the classical synthesis problem, i.e., forcing the closed loop system to converge to $\mathrm y^\star$ (see Theorem \ref{thm.ClosedLoopSynthesis}). As we noted, if $\mathrm w\in \mathbb{R}^{|\EE_\G|}$ lies in the kernel of $\E_\G P_\G$, then the solution of the following equations is the same:
\begin{align*}
k^{-1}(\mathrm y) + \E_\G \gamma (\E_\G^\top  \mathrm y) &= -\E_\G P_\G \mathrm w,~~ k^{-1}(\mathrm y) + \E_\G \gamma (\E_\G^\top  \mathrm y) = 0.
\end{align*}
Thus, if $\mathrm w$ lies in $\ker(\E_\G P_\G)$, running the interaction protocol $(\Pi,\mathrm w)$ does not change the steady-state output of the system. However, by Theorem \ref{thm.SynthesisRandomIndication}, a random vector in $\ker(\E_\G P_\G)$ gives a $(\mathfrak G,\G)$-edge-indication vector, as long as $\G$ is $2$-connected. Thus, using the interaction protocol $(\Pi,\mathrm w)$, where $\mathrm w \in \ker \E_\G P_\G$ is chosen randomly, will guarantee that all faulty system would converge to a steady-state different than $\mathrm y^\star$, with probability $1$.  Thus, applying the algorithm $\mathscr{A}$ allows an on-line, finite time distinction between the nominal faultless system and its faulty versions. We explicitly write the prescribed algorithm below:

\begin{algorithm} 
\caption{Network Fault Detection in MEIP MAS}
\label{alg.FaultDetection}
\begin{algorithmic}[1]
\State\label{step.FaultDetection1}Find edge controllers $\{{\Pi}_e\}_{e\in \mathbb{E}_\G}$ solving the synthesis problem with graph $\G$, agents $\Sigma$, and control goal $\mathrm y^\star$ (see Theorem \ref{thm.ClosedLoopSynthesis} and Remark \ref{rem.FormReconfig}).
\State Find a basis $\{b_1,...,b_k\}$ for the linear space $\ker \E_\G P_\G$
\State Pick a Gaussian vector $\alpha \in \mathbb{R}^k$
and define $\mathrm w = \sum_i^k \alpha_i b_i$
\State \label{step.FaultDetection5}Define the interaction protocol as $({\Pi},\mathrm w)$.
\State Run the system with the chosen interaction protocol.
\State Implement {the algorithm $\mathscr{A}$ for the system $(\G,\Sigma,(\Pi,\mathrm w))$ with limit $\mathrm y^\star$. Declare a fault in the network if $\mathscr{A}$ declares that the system does not converge to the prescribed value.}
\end{algorithmic}
\end{algorithm}

\begin{thm} [{Network} Fault Detection] \label{ThmAboutFaultDetection}
Suppose that $n$ agents $\{\Sigma_i\}$ and {an underlying} graph $\G$ are given, and that $\{\Sigma_i\}$ satisfy Assumption \ref{Assumption}. Suppose furthermore that $\G$ is 2-connected. Then, {with probability 1}, Algorithm \ref{alg.FaultDetection} synthesizes an interaction protocol $(\Pi,\mathrm w)$ solving Problem \ref{prob.NFD}, i.e., the algorithm satisfies the following properties:
\begin{itemize}
\item[i)] If no faults occur in the network, the output of the closed-loop system converges to $\mathrm y^\star$.
\item[ii)] {If any number of faults occur in the network, the algorithm detects their existence}.
\end{itemize}
\end{thm}
\begin{proof}
Follows from the discussion preceding Algorithm \ref{alg.FaultDetection}. Namely, Theorem \ref{thm.SynthesisRandomIndication} assures that $\mathrm w$ is a $(\mathfrak G,\G)$-edge-indication vector, so long that $\G$ is 2-connected. In other words, the output of the closed-loop system with graph $\G$ converges to $\mathrm y^\star$, and for any graph $\G \neq \HH \in \mathfrak G$, the output of the closed-loop system with graph $\HH$ converges to a value different from $\mathrm y^\star$.

It remains to show that the algorithm declares a fault if and only if a fault occurs. If no faults occur, $\mathscr{A}$ never declares a fault, and the same is true for Algorithm \ref{alg.FaultDetection}. On the contrary, assume any number of faults occur in the network, and let $\HH$ be the current underlying graph. The output of the closed loop system converges to $\mathrm y \neq \mathrm y^\star$, so $\mathscr{A}$ eventually stops and declares that the network does not converge to the conjectured limit. Thus Algorithm \ref{alg.FaultDetection} declares a fault.
\end{proof} 
\subsection{Multi-Agent Synthesis in the Presence of an Adversary}\label{subsec.Adversary}
Consider the following 2-player game. Both players are given the same $n$ SISO agents $\Sigma_1,\cdots,\Sigma_n$, the same graph $\G$ on $n$ vertices and $m$ edges, and the same vector $\mathrm y^\star \in \R^n$. There is also a server that can measure the state of the agents at certain intervals, and broadcast a single message to all agents once. The planner acts first, and designs a control scheme for the network and the server. The adversary acts second, removing at most $r$ edges from $\G$. The system is then run. The planner wins if the closed-loop system converges to $\mathrm y^\star$, and the adversary wins otherwise. We show that the planner can always win by using a strategy stemming from edge-indication vectors, assuming the agents are MEIP.

Namely, consider the following strategy. Take all possible $\sum_{\ell=0}^r \binom{m}{\ell}$ underlying graphs. For each graph, the planner solves the synthesis problem as in Theorem \ref{thm.ClosedLoopSynthesis}. If the planner finds out the adversary changed the underlying graph to $\HH$, he could notify the agents of that fact (through the server), and have them run the protocol solving the synthesis problem for $\HH$. Thus the planner needs to find a way to identify the underlying graph after the adversary took action, without using the server's broadcast. This can be done by running the system with a $\mathfrak{G}$-edge-indication vector, and using the server to identify the network's steady-state. Namely, consider Algorithms \ref{alg.Adversarial}, \ref{alg.AdversarialInRunAgents} and \ref{alg.AdversarialInRunServer}, detailing the synthesis procedure and in-run protocol for the planner. We prove:

\begin{algorithm} [t]
\caption{Planner Strategy in Adversarial Multi-Agent Synthesis Game with MEIP agents - Synthesis}
\label{alg.Adversarial}
\begin{algorithmic}[1]
\State Define $N = \sum_{\ell=0}^{r} \binom{m}{\ell}$. Let $\mathrm{Graphs}$ be an array with $N$ entries, and let $j = 1$. \;
\For {$\ell = 0,\cdots,r$}
\For {$1\le i_1 < i_2 < \cdots < i_\ell \le m$}
\State Insert the graph $\HH = \G - \{e_{i_1},\cdots,e_{i_\ell}\}$ to the $j$-th entry of $\mathrm{Graphs}$. Advance $j$ by $1$.
\EndFor
\EndFor
\State Define two arrays $\mathrm{Controllers}$,$\mathrm{SSLimits}$ of length $N$.
\State Choose $\mathrm w$ as Gaussian random vector of length $m$.
\For {$j = 1,\cdots,N$}
\State Take any edge controllers $\{\Pi_e\}_{e\in\mathbb{E}_\G}$ satisfying Assumption \ref{Assumption}.
\State Compute the steady-state limit of the network with agents $\Sigma$, underlying graph $\mathrm{Graphs}(j)$, and interaction protocol $(\Pi,\mathrm w)$. Insert the result into $\mathrm{SSLimits}(j)$
\State Solve the synthesis problem for agents $\Sigma$ and underlying graph $\mathrm{Graphs}(j)$ as in Theorem \ref{thm.ClosedLoopSynthesis} and Remark \ref{rem.FormReconfig}. Insert the result into $\mathrm{Controllers}(j)$
\EndFor
\end{algorithmic}
\end{algorithm}

\begin{algorithm} [b]
\caption{Planner Strategy - In-Run Protocol for Agents}
\label{alg.AdversarialInRunAgents}
\begin{algorithmic}[1]
\State Run the interaction protocol $(\Pi,\mathrm w)$. 
\State When a message $j$ is received, run the interaction protocol described by $\mathrm{Controllers}(j)$.
\end{algorithmic}
\end{algorithm}

\begin{algorithm} [t]
\caption{Planner Strategy - In-Run Protocol for Server}
\label{alg.AdversarialInRunServer}
\begin{algorithmic}[1]
\State Define $\text{HasFaulted}$ as an array of zeros of size $N$.
\While{\text{HasFaulted has at least two null entries}}
\State Run $N$ instances of the {algorithm $\mathscr{A}$} simultaneously, with conjectured steady-states limits from $\text{SSLimits}$.
\For {$j=1$ to $N$}
\If {The $j$-th instance declared ``no"}
\State Change the value of $\text{HasFaulted}(j)$ to $1$.
\EndIf
\EndFor
\EndWhile
\State Find {the} index $j$ such that $\text{HasFaulted}(j) = 0$. Broadcast the message $j$ to the agents.
\end{algorithmic}
\end{algorithm}

\begin{thm} \label{thm.Adversarial}
Consider the game above. {With probability 1,} Algorithms \ref{alg.Adversarial}, \ref{alg.AdversarialInRunAgents} and \ref{alg.AdversarialInRunServer} describe a winning strategy for the planner. Moreover, if $r$ is independent of $n$ (i.e., $r = O(1)$), the synthesis algorithm has polynomial time complexity. Otherwise, the time complexity is $O(n^{cr})$ for some universal constant $c>0$. Furthermore, the size of the message broadcasted by the server is $O(r\log n)$.
\end{thm}
\begin{proof}
Suppose the adversary changed the underlying graph to $\HH$, which has entry $j$ in $\text{Graphs}$. We note that Assumption \ref{assump.Diff} assures that {$\mathscr{A}$} never declares a fault if and only if the closed-loop system converges to the conjectured steady-state, and that $\mathrm w$ is a $\mathfrak G$-edge-indication vector by Theorem \ref{thm.GeneralRandomIndication}. Thus, the $j$-th instance of convergence assertion protocol can never return a fault, and all other instances must eventually declare a fault. Thus, the server correctly guesses the underlying graph. It then broadcasts the index $j$ to the agents, allowing them to change the interaction protocol and run the solution of the synthesis problem with desired output $\mathrm y^\star$ and underlying graph $\HH$. Thus the network will converge to $\mathrm y^\star$, and the planner wins.

We now move to time complexity. Note that $N = O(m^{r})$. The first for-loop has $N$ iterations, each of takes no more than $O(mn)$ actions done (where a graph is saved in memory by its incidence matrix, which is of size $\le m\times n$). Thus the first for-loop takes a total of $O(m^{r+1}n)$ time.  The second for-loop is a bit more complex. It solves the synthesis problem for $\HH$ by solving an equation of the form $\E_\HH \mathrm v = \mathrm v_0$ for some known vector $\mathrm v_0$ and an unknown $\mathrm v$. This can be done using least-squares, which takes no more than $O(\max\{m,n\}^3)$ time. As for finding the steady-state, this can be done by solving a convex network optimization problem \cite[Problem (OPP)]{SISO_Paper}, which takes a polynomial amount of time in $n,m$ (e.g. via gradient descent). Recalling that $m \le \binom{n}{2} = O(n^2)$, we conclude that if $r$ is bounded, the total time used is polynomial in $n$. Moreover, if $r$ is unbounded, the bottleneck is the first for-loop which takes $O(m^{r+1}n)$ time. Plugging $m \le n^2$ gives a bound on the time complexity of the form $O(n^{2r})$, as $n^3 = O(n^{2r})$. As for the communication complexity,the server broadcasts a number between $1$ and $N$, so a total of $O(\log_2 N)$ bits are needed to transmit the message. Plugging in $N = O(m^{r})$ gives that the number of bits needed is $O(r\log_2 m) = O(r\log n)$.
\end{proof}

\subsection{Detection and Isolation of Network Faults} \label{subsec.FI}
We now consider Problem \ref{prob.NFI}, in which faults occur throughout the run, and we want to detect their existence and overcome them. This problem can be thought of as a tougher hybrid of the previous two problems - in subsection \ref{subsec.FaultDetection}, the faults could appear throughout the run, but we only needed to find out they exist. In subsection \ref{subsec.Adversary}, all of the faults occur before the run, but we had to overcome them. Motivated by this view, we offer a hybrid solution. Ideally, the interaction protocol will have two disjoint phases - a first, ``stable" phase in which the underlying graph is known and no extra faults have been found, and a second, ``exploratory" phase in which extra faults have been found, and the current underlying graph is not yet known. The first phase can be solved by using the {network} fault detection Algorithm \ref{alg.FaultDetection}, as long as the current underlying graph is 2-connected. The second phase can be solved by the pre-broadcast stage of the planner strategy described in Algorithms \ref{alg.Adversarial}, \ref{alg.AdversarialInRunAgents}, and \ref{alg.AdversarialInRunServer}. 

The main issue that remains is what happens if the underlying graph changes again in the exploratory phase, i.e. we entered the exploratory phase with underlying graph $\HH_1$, but it changed to $\HH_2$ before we identified that graph. In the exploratory phase, we run an instance of $\mathscr{A}$ on all of the possible graphs simultaneously, until all but one instance declared a fault. If the instance related to $\HH_2$ has not declared a fault yet, it will not declare a fault from now on, unless another fault occurs before the end of the exploratory phase. If the same instance has already declared a fault, we have a problem - all other instances will eventually also declare a fault, and there are two options in this case.

The first option is that one instance will declare a fault last, i.e. there is a time in which all but one instance have declared a fault. In this case, we identify the graph as some $\HH_3$. When we return to the stable phase and run the interaction protocol $(\Pi,\mathrm w)$ corresponding to $\HH_3$, a fault will be declared and we will return to the exploratory phase. This is because $\mathrm w$ was synthesized as a $(\mathfrak G, \HH_3)$-edge-indication vector, meaning the de-facto steady-state limit (with graph $\HH_2$) will be different than the conjectured steady-state limit (with graph $\HH_3$), and $\mathscr{A}$ will declare a fault. The second option is that the last few instances of the convergence assertion protocol declare a fault simultaneously, i.e. there is a time in which all instances have declared a fault, which is dealt with by restarting the exploratory phase. We get the synthesis algorithm and in-run protocol presented in Algorithms \ref{alg.FaultDetectionSynthesis} and \ref{alg.FaultDetectionInRun}. We claim these solve Problem \ref{prob.NFI}. We prove:

\begin{algorithm} [t]
\caption{Synthesis for Network Fault Isolation}
\label{alg.FaultDetectionSynthesis}
\begin{algorithmic}[1]
\State Define $N = \sum_{\ell=0}^{r} \binom{m}{\ell}$. Let $\mathrm{Graphs}$ be an array with $N$ entries, and let $j = 1$ \;
\For {$\ell = 0,\cdots,r$}
\For {$1\le i_1 < i_2 < \cdots < i_\ell < m$}
\State Insert the graph $\HH = \G - \{e_{i_1},\cdots,e_{i_\ell}\}$ to the $j$-th entry of $\mathrm{Graphs}$. Advance $j$ by $1$.
\EndFor
\EndFor
\State Define two arrays $\mathrm{IP}$,$\mathrm{SSLimits}$ of length $N$.
\State Choose $\mathrm w$ as a Gaussian random vector of length $m$.
\State Choose controllers $\{\Pi_e\}_{e\in\EE}$ satisfying Assumption \ref{Assumption}.
\For {$j = 1,\cdots,N$}
\State Run steps \ref{step.FaultDetection1}-\ref{step.FaultDetection5} of Algorithm \ref{alg.FaultDetection}. Insert the resulting interaction protocol into $\mathrm{IP}(j)$.
\State Compute the steady-state limit of the closed-loop system with the interaction protocol $(\Pi,\mathrm w)$. Insert the result into $\mathrm{SSLimits}(j)$
\EndFor
\end{algorithmic}
\end{algorithm}

\begin{algorithm} [t]
\caption{In-Run Protocol for Network Fault Isolation}
\label{alg.FaultDetectionInRun}
\begin{algorithmic}[1]
\State Find the index $j$ for which $\text{Graphs}(j) = \G$.
\State \label{step.InRun0}Command the agents to change their interaction protocol to the one described in $\text{IP}(j)$. Define $\HH = \text{Graphs}(j)$.
\State \label{step. InRun0end} Run $\mathscr{A}$ for the closed-loop system with graph $\HH$ and interaction protocol $\mathrm{IP}(j)$. Only if the algorithm declares a fault, continue to step \ref{step.InRun1}.
\State \label{step.InRun1}Define $\text{HasFaulted}$ as an array of zeros of size $N$.
\State Change the agents' interaction protocol to $(\Pi,\mathrm w)$.
\While{\text{HasFaulted has at least two null entries}}
\State Run $N$ instances of the convergence assertion protocol {from Assumption \ref{assump.Diff}} simultaneously, with conjectured limits from $\text{SSLimits}$.
\For {$j=1$ to $N$}
\If {The $j$-th instance has declared a fault}
\State Change the value of $\text{HasFaulted}(j)$ to $1$.
\EndIf
\EndFor
\EndWhile \label{step.InRun1end}
\If{\text{HasFaulted} has no entries equal to zero}
\State Go to step \ref{step.InRun1}.
\EndIf
\State Find {the} index $j$ such that $\text{HasFaulted}(j) = 0$. Set $\HH = \text{Graphs}(j)$. Go to step \ref{step.InRun0}.
\end{algorithmic}
\end{algorithm}

\begin{thm}
Let $\Sigma_1,\cdots,\Sigma_n$ be agents satisfying Assumption \ref{Assumption}, and let $\G$ be a $k$-connected graph for $k\ge 3$ on $n$ vertices and $m$ edges. Then, {with probability 1}, Algorithms \ref{alg.FaultDetectionSynthesis} and \ref{alg.FaultDetectionInRun}, run with $r = k-2$, solve Problem \ref{prob.NFI} for up to $r$ faults.
\end{thm}

\begin{proof}
We refer to steps \ref{step.InRun0} to \ref{step. InRun0end} of Algorithm \ref{alg.FaultDetectionInRun} as the stable phase of the algorithm, and to steps \ref{step.InRun1} to \ref{step.InRun1end} as the exploratory phase. As the number of faults is no bigger than $r=k-2$, the underlying graph remains 2-connected throughout the run.
We claim the theorem follows from the following simple claims:
\begin{enumerate}
\item \label{fact.1} If we are in the stable phase, the current underlying graph is $\HH = \mathrm{Graphs}(j)$, and no more faults occur throughout the run, the closed-loop system converges to $\mathrm y^\star$.
\item \label{fact.2} If we are in the stable phase, but the assumed graph $\HH = \mathrm{Graphs}(j)$ is not the current underlying graph, we will eventually move to the exploratory phase.
\item \label{fact.4} Each instance of the exploratory phase eventually ends.
\item \label{fact.5} If an instance of the exploratory phase is executed, and no more faults occur throughout the run, it correctly identifies the current underlying graph.
\end{enumerate}
We first explain why the claims hold with probability 1. Claims \ref{fact.1} and \ref{fact.2} follow from Theorem \ref{ThmAboutFaultDetection} with probability 1, as the underlying graph $\G$ is always 2-connected. Claim \ref{fact.5} holds with probability 1, as follows from Theorem \ref{thm.Adversarial}. As for Claim \ref{fact.4}, if no faults occur throughout the instance of the exploratory phase, then it eventually ends by Claim \ref{fact.5}. If faults do occur throughout the run, the discussion above shows that at some time, all (except possibly one) instances of the convergence assertion protocol have declared a fault. If all of them declared a fault, we start another instance of the exploratory phase, and otherwise, we move to the stable phase. In either case, the instance of the exploratory phase ends, and Claim \ref{fact.4} is proved.

We now explain how the theorem follows from these claims. Suppose a total of $\ell \le r$ faults occur throughout the run. Let $T<\infty$ be the time at which the last fault occurs. We look at the phase of the algorithm at times $t>T$, and show that in either case the system must converge to $\mathrm y^\star$.
\begin{itemize}
\item[i)] If we arrive at the stable phase with and the conjectured graph $\mathcal{H}$ is the true underlying graph, then the system converges to $\mathrm y^\star$ (Claim \ref{fact.1}). 
\item[ii)] If we start an instance of the exploratory phase, it eventually ends (Claim \ref{fact.4}) and the stable phase starts, in which the conjectured graph $\mathcal{H}$ is the true underlying graph (Claim \ref{fact.5}). By i), the system converges to $\mathrm y^\star$.
\item[iii)] If we are in the stable phase, but the conjectured graph $\mathcal{H}$ is different from the true underlying graph, we eventually start an exploratory phase (Claim \ref{fact.2}). Thus, the system converges to $\mathrm y^\star$ by ii).
\item[iv)] Lastly, we could be in the middle of an instance of the exploratory phase. In that case, the instance eventually ends (Claim \ref{fact.4}), after which we either apply a new instance of the exploratory phase, or the stable phase. In both cases, we can use either i), ii) or iii) and conclude the system must converge to $\mathrm y^\star$. 
\end{itemize}
\if(0)
We now prove the claims. Note that if $\HH$ can be yielded by removing no more than $r$ edges from $\G$, then it is $2$-connected. Indeed, if the removed edges are $e_1,...,e_\ell$, choose a vertex $v_i$ for each of them, so that $\HH$ contains the graph $\HH_1 = \G - \{v_1,\cdots,v_\ell\}$. For any vertex $v\neq v_1,\cdots,v_\ell$, because $\ell \le r$ and $\G$ is $(r+2)$-connected, $\HH_1 - \{v\} = \G - \{v_1,\cdots.v_\ell,v\}$ is connected. Thus $\HH_1$, and hence $\HH$, is 2-connected.  

By Theorem \ref{thm.SynthesisRandomIndication}, for each $j$, the vector $\mathrm w_j$ from the interaction protocol $\mathrm{IP}(j)$ is a $(\mathfrak G,\text{Graphs}(j))$-edge-indication vector. As all graphs achieved by removing no more than $r$ edges from $\G$ are non-empty, we conclude that for every $j_1,j_2$, if the system is run with interaction protocol $\mathrm{IP}(j_1)$, the system with underlying graph $\mathrm{Graphs}(j_1)$ will converge to a different value from the system with underlying graph $\mathrm{Graphs}(j_2)$. We thus conclude by {Assumption \ref{assump.Diff}}, that claims \ref{fact.1} and \ref{fact.2} are true.

We now prove claims \ref{fact.4} and \ref{fact.5}. By Theorem \ref{thm.GeneralRandomIndication}, the chosen vector $\mathrm w$ is a $\mathfrak G$-edge-indication vector. Thus, for any two different graphs $\mathrm{Graphs}(j_1)$ and $\mathrm{Graphs}(j_2)$, the steady-state output will be different. Thus, by {Assumption \ref{assump.Diff}}, at least $N-1$ instances of {$\mathscr{A}$} must eventually declare a fault (as there is only one true steady-state), even if the underlying graph changed while running this phase. This proves claim \ref{fact.4}. Moreover, suppose that the last fault of the run happened before we started executing this instance of the exploratory phase. The true underlying graph appears in $\mathrm{Graphs}$, as it is achieved from $\G$ by removing no more than $r$ edges. If it the true underlying graph is equal to $\mathrm{Graphs}(j)$, then by {Assumption \ref{assump.Diff}}, the $j$-th instance of the convergence assertion method will never declare a fault. Thus, the last remaining non-zero entry of $\mathrm{HasFaulted}$ is the $j$-th, meaning that we correctly identify the current underlying graph. This proves claim \ref{fact.5} and completes the proof of the theorem.
\fi
\end{proof}
\vspace{-10pt}
\begin{rem}
We can use a similar protocol to isolate more complex faults. We consider the collection of subgraphs $\HH$ of $\G$ in which there is a set of vertices of size $\le r$, so that each edge in $\G-\HH$ touches at least one vertex in the set.  This observation allows us to offer similar network fault detection and isolation algorithms for more complex types of faults. For example, we can consider a case in which each agent communicates with all other agents by a single transceiver, and if it faults, then all edges touching the corresponding vertex are removed from the graph. We can even use a hybrid fault model, in which faults correspond to certain subsets of edges touching a common vertex are removed from the graph. For example, suppose there are two distant groups of agents. Agents in the same group are close, and communicate using Bluetooth communication. Agents in different groups are farther, and communicate using Wi-Fi (or broadband cellular communication). When an agent's Bluetooth transceiver faults, all inter-group edges are removed, and when the Wi-Fi transceiver faults, all intra-group edges are removed.
\end{rem}

\section{Online Assertion of Network Convergence} \label{sec.ConvAssert}
In the previous section, we used the notion of edge-indication vectors, together with Assumption \ref{assump.Diff}, to suggest algorithms for network fault detection and isolation. The goal in this section is to propose algorithms $\mathscr{A}$ satisfying Assumption \ref{assump.Diff}. This will be achieved by using convergence estimates, relying on passivity. 
We revisit a result from \cite{SISO_Paper}.
\begin{prop}[\cite{SISO_Paper}]
Let $(\mathrm{u,y,\zeta,\mu})$ be a steady-state of $(\G,\Sigma,\Pi)$ {of the form \eqref{AgentsandControllers}}. Suppose that the agents $\Sigma_i$ {, with state $x_i$,} are passive with respect to $(\mathrm{u_i,y_i})$ with passivity index $\rho_i \ge 0$, and that the controllers $\Pi_e$, with state $\eta_e$,  are passive with respect to $(\mathrm{\zeta_e,\mu_e})$, with passivity index $\nu_e\ge 0$. Let $S_i(x_i)$ and $W_e(\eta_e)$ be the agents' and the controllers' storage functions. Then $S(x,\eta) = \sum_{i\in \V} S_i(x_i) + \sum_{e\in \EE} W_e(\eta_e)$ is a positive-definite $C^1$-function, which nulls only at the steady-states $(\mathrm{x,\eta})$ {corresponding to $(\mathrm{u_i,y_I})$ and $(\mathrm{\upzeta_e,\upmu_e})$}, and satisfying the inequality:
\begin{align} \label{eq.PassivityInequality1}
\frac{dS}{dt} \le -\sum_{i\in \V} \rho_i(y_i(t)-\mathrm y_i)^2 - \sum_{e\in \EE} \nu_e(\mu_e(t) - \upmu_e)^2.
\end{align}
\end{prop}
\begin{proof}
The proof follows immediately from $S_i,W_e$ being positive-definite $C^1$-functions nulling only at $\mathrm x_i, \mathrm \eta_e$, by summing the following inequalities:
\begin{align*} 
\frac{dS_i}{dt} &\le (u_i(t)-\mathrm u_i)(y_i(t)-\mathrm y_i) - \rho_i(y_i(t)-\mathrm y_i)^2 \\
\frac{dW_e}{dt} &\le (\mu_e(t)-\mathrm \upmu_e)(\zeta_e(t)-\mathrm \upzeta_e) - \nu_e(\mu_e(t) - \upmu_e)^2,
\end{align*}
and using the equality $(u(t)-\mathrm u)^\top (y(t)-\mathrm y) = -(\mu(t)-\mathrm \upmu)^\top \E_\G^\top (y(t)-\mathrm y) = -(\mu(t)-\mathrm \upmu)(\zeta(t) - \mathrm \upzeta)$.
\end{proof}

The inequality \eqref{eq.PassivityInequality1} can be thought of as a way to check that the system is functioning properly. Indeed, we can monitor $x$, $y$, $\eta$, and $\mu$, and check that the inequality holds. If it does not, there must have been a fault in the system. This idea has a few drawbacks, linked to one another.
First, {as we commented in Subsection \ref{subsec.Assumptions}, in some networks, the controller state $\eta_e(t)$ can be defined only for existing edges, so using $\eta(t)$ requires us to know the functioning edges, which is absurd. Thus, in some cases, we must use} Assumption \ref{Assumption_g_1}. Second, in practice, even if we have access to $x$, we cannot measure it continuously. Instead, we measure it at certain time intervals. One can adapt \eqref{eq.PassivityInequality1} to an equivalent integral form:
\begin{align} \label{eq.PassivityInequality}
&S(x(t_{k+1}),\eta(t_{k+1}))-S(x(t_k),\eta(t_k)) \le \\& -\int_{t_{k}}^{t_{k+1}}(\sum_{i\in \V} \rho_i\Delta y_i(t)^2 + \sum_{e\in \EE} \nu_e\Delta \mu_e(t)^2)dt, \nonumber
\end{align}
where $\Delta y_i = y_i(t) - \mathrm y_i$ and $\Delta \mu_e = \mu_e(t) - \upmu_e$. However, this gives rise to the third problem - unlike the function $S$, we can not assure that the functions $(y_i(t) - \mathrm y_i)^2$ and $(\mu_e(t) - \upmu_e)^2$ (or their sum) is monotone. Thus, we cannot compute the integral appearing on the right-hand side of the inequality. 

We present two approaches to address this problem. First, we try and estimate the integral using high-rate sampling, by linearizing the right hand side of \eqref{eq.PassivityInequality} and bounding the error. Second, we try to bound the right-hand side as a function of $S$, resulting in an inequality of the form $\dot{S}\le-\mathcal F (S)$, which will give a convergence estimate.

\subsection{Asserting Convergence Using High-Rate Sampling}
Consider the inequality \eqref{eq.PassivityInequality}, and suppose $t_{k+1}-t_k = \Delta t_k$ is very small, so the functions $y_i(t) - \mathrm y_i$ and $\mu_e(t) - \upmu_e$ are roughly constant in the time period used for the integral. More precisely, recalling that $y = h(x)$ and $\mu = \phi(\eta,\E_\G^\top  y)$, and assuming these functions are differentiable near $x(t_k),\eta(t_k)$, we expand the right-hand side of \eqref{eq.PassivityInequality} to a Taylor series,
\begin{align} \label{eq.TaylorPassivity}
&\int_{t_{k}}^{t_{k+1}}\left(\sum_{i\in \V} \rho_i\Delta y_i(t)^2 + \sum_{e\in \EE} \nu_e\Delta \mu_e(t)^2\right)dt = \\
&\left(\sum_{i\in\V} \rho_i \Delta y_i (t_{k})^2 + \sum_{e\in \EE} \nu_e\Delta\mu_e(t_{k})^2\right)\Delta t_k + O(\Delta t_k^2) .\nonumber
\end{align}

We wish to give a more explicit bound on the $O(\Delta t_k^2)$ term. We consider the following function $G$, defined on the interval $[t_k,t_{k+1}]$ by the formula 
\begin{equation}\label{Gsimp}
G(t) = \sum_i \rho_i (y_i(t) - \mathrm y_i)^2 + \sum_e \nu_e(\mu_e(t) - \upmu_e)^2.
\end{equation}
The equation \eqref{eq.TaylorPassivity} is achieved from the approximation $G(t)=G(t_{k})+O(|t-t_{k}|)$ which is true for differentiable functions. Using Lagrange's mean value theorem for $t\in[t_k,t_{k+1}]$, we find some point $s\in (t,t_{k+1})$ such that $G(t) = G(t_{k}) + \frac{dG}{dt}(s)(t-t_{k})$. If we manage to bound the time derivative $\frac{dG}{dt}$ in the interval $[t_k,t_{k+1}]$, we would find a computational way to assert convergence. By the chain rule, the time derivative of $G$ is given by
\begin{align}\label{eq.DerivativeOfG}
\frac{dG}{dt} = \sum_{i\in\V} \rho_i (y_i(t) - \mathrm y_i)\dot{y}_i + \sum_{e\in\EE} \nu_e (\mu_e(t) - \upmu_e)\dot{\mu}_e.
\end{align}
In order to compute the time derivative of $y_i,\mu_i$, we recall that both are functions of $x$ and $\eta$, namely $y = h(x)$ and $\mu = \phi(\eta,\E_\G^\top  y)  = \phi(\eta,\E_\G^\top  h(x))$. Thus, we have that\small
\begin{align} \label{eq.dotsOfSignals}
\begin{cases}
\dot{y} &= \nabla_x h(x(t))\dot{x}\\  
\dot{\mu} &= \nabla_\eta \phi(\eta(t),\zeta(t))\dot{\eta}+\nabla_x \phi(\eta(t),\zeta(t))\E_\G^\top \nabla h(x(t))\dot{x} ,
\end{cases}
\end{align}\normalsize
where $\zeta(t) = \E_\G^\top h(x(t))$,  $\dot{x} = f(x,u) = f(x,-\E_\G\phi(\eta,\zeta))$, and $\dot{\eta} = \psi(\eta,\zeta) = \psi(\eta,\E_\G^\top h(x))$.
\if(0)
where $\dot{x} = -f(x) + u = -f(x) + \E_\G g(\E_\G^\top h(x))$.\fi Thus we can write the time derivative of $G$ as a continuous function of $x(t),\eta(t)$, as we plug the expressions for $\dot{y},\dot{\mu}$ into \eqref{eq.DerivativeOfG}. However, we do not know the value of $x(t),\eta(t)$ between measurements. 
To tackle this problem,notice that we have some information on where $x(t),\eta(t)$ can lie. Namely, equation \eqref{eq.PassivityInequality1} shows that $S(x(t),\eta(t))$ is a monotone descending function. Thus, we know that $x(t),\eta(t)$ lie in the set $\mathbb B = \{x,\eta:\ S(x,\eta) \le S(x(t_k),\eta(t_k))\}$. More precisely, we show the following.

\begin{prop} \label{prop.BoundDerivativeG}
Assume the functions $h_i,f_i,\phi_e,\psi_e$ are all continuously differentiable. Then for any time $t\in [t_k, t_{k+1}]$, the following inequality holds:\small
\begin{align*}
\bigg|\frac{dG}{dt}(t)\bigg| \le (\rho_\star M_{\Delta y}M_{\dot{y}} + \nu_\star M_{\Delta \mu}M_{\dot{\mu},x})M_{\dot{x}} + \nu_\star M_{\Delta \mu} M_{\dot{\mu},\eta}M_{\dot{\eta}},
\end{align*}\normalsize
where $
M_{\dot{x}} = \max_{(x,\eta)\in\mathbb B} \|f(x,-\E_\G\phi(\eta,\E_\G^\top h(x))\|,~$$
M_{\dot{\eta}} = \max_{(x,\eta)\in\mathbb B} \|\psi(\eta,\E_\G^\top h(x))\|,~$$
M_{\dot{y}} = \max_{(x,\eta)\in \mathbb{B}} \|\nabla_x h(x)\|,~$
$M_{\dot{\mu},x} = \max_{(x,\eta)\in \mathbb{B}} \|\nabla_\zeta \phi(\eta,\E_\G^\top h(x)) \E_\G^\top  \nabla_x h(x)\|,$
$
M_{\dot{\mu},\eta} = \max_{(x,\eta)\in \mathbb{B}} \|\nabla_\eta \phi(\eta,\E_\G^\top h(x))\|,$
$
M_{\delta y} = \max_{(x,\eta)\in \mathbb B} \|h(x) - h(\mathrm x)\|,~$
$
M_{\delta \mu} = \max_{(x,\eta)\in \mathbb B} \|\psi((\eta,\E_\G^\top h(x)) -\upmu\|,~$
$
$, $\rho_\star = \max_i \rho_i$, $\nu_e = \max_e \nu_e$, and $\mathbb B = \{(x,\eta):\ S(x,\eta) \le S(x(t_k),\eta(t_k))\}$.
\end{prop}

\begin{proof}
We fix some $t \in [t_k,t_{k+1}]$, so that $(x(t),\eta(t)) \in \mathbb B$. We use the expressions for $\dot{x},\dot{\eta},\dot{y},\dot{\mu}$ found {in \eqref{eq.dotsOfSignals}}. First, the conditions $\|\dot{x}\| \le M_{\dot{x}}$ and $\|\dot{\eta}\| \le M_{\dot{\eta}}$ are obvious. Equation \eqref{eq.dotsOfSignals} shows that $\|\dot{y}\| \le M_{\dot{y}} M_{\dot{x}}$ and $\|\dot{\mu}\| \le M_{\dot{\mu},x}M_{\dot{x}}+M_{\dot{\mu},\eta}M_{\dot{\eta}}$. By using Cauchy-Schwartz inequality on \eqref{eq.DerivativeOfG}, we obtain $\left|\frac{dG}{dt}\right| \le \rho_\star M_{\delta y}\|\dot{y}\| + \nu_\star M_{\delta \mu} \|\dot{\mu}\|$, concluding the proof.
\end{proof}

\if(0)
\begin{rem}
Suppose that Assumption \ref{Assumption_g_2} holds, so that the agents are given by $\dot{x}_i = -f_i(x_i) + q_i(x_i)u_i; y_i = h_i(x_i)$ and the controllers are given by $\mu = g(\zeta)$. In that case, $\dot{x} = -f(x) + u = -f(x) - \E_\G g(\E_\G^\top  x)$ and $\mu = g(\E_\G^\top h(x))$, so we can get a slightly more comprehensive bound by applying the same analysis. Namely, 
\begin{align*}
\bigg|\frac{dG}{dt}(t)\bigg| \le (\rho_\star M_{\delta y}M_{\dot{y}} + \nu_\star M_{\delta \mu}M_{\dot{\mu}})(M_q M_u+M_f)
\end{align*}
where $
M_u = \max_{x\in \mathbb B} ||\E_\G g(\E_\G^\top h(x))||,~$$
M_f = \max_{x\in \mathbb B} ||f(x)||,~$$
M_q = \max_{x\in \mathbb B} \max_{i\in\V} |q_i(x)| ,~$$
M_{\delta y} = \max_{x\in \mathbb B} ||h(x) - h(\mathrm x)|| ,~$$
M_{\delta \mu} = \max_{x\in \mathbb B} ||g(\E_\G^\top h(x)) - g(\E_\G^\top h(\mathrm x))|| ,~$$
M_{\dot{y}} = \max_{x\in \mathbb{B}} ||\nabla h(x)|| ,~$$
M_{\dot{\mu}} = \max_{x\in \mathbb{B}} ||\nabla g(\E_\G^\top h(x)) \E_\G^\top  \nabla h(x)|| ,~
$, $\rho_\star = \max_i \rho_i$, $\nu_e = \max_e \nu_e$, and $\mathbb B = \{x:\ S(x) \le S(x(t_k))\}$.
\end{rem}
\fi

\begin{cor} \label{cor.IneqSatisfied}
Fix any two times $t_k < t_{k+1}$, and consider the notation of Proposition \ref{prop.BoundDerivativeG}. Then the following inequality holds:
\begin{align} \label{eq.HighRateSampling}
&S(x(t_{k+1})) - S(x(t_k)) \le\\& -\left(\sum_i \rho_i \Delta y_i (t_{k})^2 + \sum_{e\in \EE} \nu_e\Delta\mu_e(t_{k})^2\right)\Delta t_k + \frac{M}{2}\Delta t_k^2,\nonumber
\end{align}
where $$M = (\rho_\star M_{\Delta y}M_{\dot{y}} + \nu_\star M_{\Delta \mu}M_{\dot{\mu},x})M_{\dot{x}} + \nu_\star M_{\Delta \mu} M_{\dot{\mu},\eta}M_{\dot{\eta}}.$$
\end{cor}

\begin{proof}
Recall that $G(t) = \sum_{i\in \V} \rho_i (y_i(t) - \mathrm y_i)^2 + \sum_{e\in \EE} \nu_e (\mu_e(t) - \upmu_e)^2$. By Proposition \ref{prop.BoundDerivativeG}, for every $t\in [t_k,t_{k+1}]$ we have $G(t) \le G(t_{k+1}) + M|t-t_{k+1}|$. Thus \eqref{eq.PassivityInequality} implies that $S(x(t_{k+1})) - S(x(t_k)) \le G(t_{k+1})\Delta t_k + \frac{M}{2}\Delta t_k^2$.
\end{proof}

The corollary proposes a mathematically-sound method for asserting convergence of the output $y(t)$ to $\mathrm y$. One samples $y(t)$, $x(t),\eta(t)$, and $\mu(t)$ at times $t_1,t_2,t_3,\ldots$. At every time instance $t_{k+1}$, one checks that the inequality \eqref{eq.HighRateSampling} holds. We show that when $\Delta t_k \to 0$, this method asserts that the output of the system converges to the said value. In other words, assuming we sample the system at a high-enough rate, we can assert it converges very closely to the supposed steady-state output. Indeed, we prove the following.

\begin{prop} \label{prop.HighRateSamplingConvergence}
Let $t_1,t_2,\cdots$, be any monotone sequence of times such that $t_k \to \infty$, and suppose that the inequality \eqref{eq.HighRateSampling} holds for any $k$. Then for any $\varepsilon > 0$, there are infinitely many $N>0$ such that $\sum_{i\in\V} \rho_i \Delta y_i (t_N)^2 + \sum_{e\in\EE} \nu_e \Delta \mu_e(t_N)^2 < \frac{M}{2}\Delta t_N+\varepsilon$. {More precisely, for any two times $t_{N_1}\le t_{N_2}$, if $t_{N_2} \ge t_{N_1} + \e^{-1}S(x(t_{N_1}),\eta(t_{N_1}))$, then there exists some $k\in\{N_1,N_1+1,\cdots, N_2\}$ such that  $\sum_{i\in\V} \rho_i \Delta y_i (t_k)^2 + \sum_{e\in\EE} \nu_e \Delta \mu_e(t_k)^2 < \frac{M}{2}\Delta t_k+\varepsilon$.}
\end{prop}

The proposition can be thought of as a close-convergence estimate. The left-hand side, viewed as a function of $x,\eta$, is a non-negative smooth function, which nulls only at the steady-state $(\mathrm x, \upeta)$. Thus it is small only when $x(t),\eta(t)$ are close to $(\mathrm {x,\upeta})$, and because we know that $S(x(t),\eta(t))$ is monotone descending, once the trajectory arrives near $(\mathrm {x,\upeta})$, it must remain near $(\mathrm {x,\upeta})$. One might ask why ``infinitely many times" is more useful in this case. Indeed, it does not add any more information if the time intervals $\Delta t_k$ are taken as constant (i.e., we sample the system at a constant rate). However, we can measure the system at an ever-increasing rate, at least theoretically. Taking $\Delta t_k \to 0$ (while still having $t_k \to \infty$, e.g. $t_k = 1/k$), we see that we must have $x(t) \to \mathrm x$ and $\eta(t) \to \mathrm \upeta$, meaning we can use the proposition to assert convergence. We now prove the proposition.
\begin{proof}
It is enough to show that for each $\varepsilon > 0$ and any $N_1>0$, there is some $N > N_1$ such that $\sum_{i\in\V} \rho_i \Delta y_i (t_N)^2 + \sum_{e\in\EE} \nu_e \Delta \mu_e(t_N)^2 < \frac{M}{2}\Delta t_N+\varepsilon$. Indeed, suppose this is not the case. Then for any $k > N_1$, the right-hand side of \eqref{eq.HighRateSampling} is upper-bounded by $-\varepsilon \Delta t_k$. We sum the telescopic series and conclude that for any $k > N_1$,\small
\begin{align} \label{eq.TelescopicSum}
S(x(t_k)) - S(x(t_{N_1})) \le -\sum_{j=N_1+1}^{k} \varepsilon \Delta t_j = -\varepsilon (t_k - t_{N_1}),
\end{align}\normalsize
so $t_{N_2} \ge t_{N_1} + \e^{-1}S(x(t_{N_1}),\eta(t_{N_1}))$ implies that $S(x(t_k),\eta(t_k)) < 0$. This is absurd, as $S \ge 0$. Thus there must exist some $N > N_1$ such that $\sum_{i\in\V} \rho_i \Delta y_i (t_N)^2 + \sum_{e\in\EE} \nu_e \Delta \mu_e(t_N)^2 < \frac{M}{2}\Delta t_N+\varepsilon$. {The second part follows from \eqref{eq.TelescopicSum} and the demand that $S(x(t_k)) \ge 0$}.
\end{proof}   

{Proposition \ref{prop.HighRateSamplingConvergence} can be used for convergence assertion.} We can consider the following scheme - begin at time $t_0$ and state $x_0,\eta_0$. We want to show that $S(x(t),\eta(t))\to 0$. We instead show that $G(t)$, defined in \eqref{Gsimp}, gets arbitrarily close to $0$. As we said, this is enough as $G(t)$ is a $C^1$ non-negative function of the state $x(t),\eta(t)$ that is only small when $x(t),\eta(t)$ is close to the steady-state $(\mathrm x,\upeta)$. We prove:

{
\begin{thm}
Consider the algorithm $\mathscr{A}$, defined in the following form. Sample the system at times $t_1,t_2,\cdots$, and check whether the inequality \eqref{eq.HighRateSampling} holds. If it does, continue, and if does not, {then} stop and declare ``no". Then there exists a sequence $t_1,t_2,\cdots$, depending on the system and the initial conditions, such that $\mathscr{A}$ satisfies assumption \ref{assump.Diff}.
\end{thm}
\begin{proof}
We present the following method of choosing $t_1,t_2,\cdots$. We first choose $t_0 = 0$, an arbitrary $\delta_1>0$, compute $M$ as in Proposition \ref{prop.BoundDerivativeG}, and choose $\Delta_1 t = \frac{\delta_1}{M}$ and $\varepsilon = \frac{\delta_1}{2}$. Sample the system at rate $\Delta_1 t$ until time $t_{N_1} > t_0 + \varepsilon^{-1}S(x_0,\eta_0)$. {The process is then reiterated with $\delta_{k+1} = \delta_1/2^k$ for $k=1,2,\cdots$, giving rates $\Delta_k t$ and times $t_{N_k}$}. We claim that $\mathscr{A}$, with this choice of sample times, satisfies Assumption \ref{assump.Diff}. If the diffusively-coupled network $(\G,\Sigma,\Pi)$ converges to $(\mathrm x,\upeta)$, then Corollary \ref{cor.IneqSatisfied} implies the algorithm never stops, as required. It remains to show that if the algorithm never stops, the network converges to the conjectured limit. By the discussion above, and the fact that $S(x(t),\eta(t))$ is a monotone descending function, it is enough to show that $\lim\inf_{k\to \infty} G(t_k) = 0$.
We first show at some point, $G(t) < \delta_1$. By choice of $\Delta_1 t$, if \eqref{eq.HighRateSampling} holds at each time, then when we reach time $t_{N_1}$, we know that at some point, we had $G(t) \le \frac{M}{2}\Delta_1 t + \epsilon = \delta_1$. Reiterating shows that at some times $t_k^\star$, $G(t_k^\star) \le \delta _k$, where $\delta_{k+1} = \frac{\delta_1}{2^k}$, so ${\lim\inf}_{k\to\infty} G(t_k) =0$.
\end{proof}
}

The term ``High-Rate Sampling" comes from the fact that if $M$ is not updated when we re-iterate with smaller $\delta$, then eventually, $t_{k+1}-t_k \to 0$, which is impractical in real-world cases. However, we note that the number $M$ decreases as $S(x(t),\eta(t))$ decreases, as shown in Proposition \ref{prop.BoundDerivativeG}. Thus, if $M$ is updated between iterations, we might have $\Delta t\not\to 0$. 

\begin{rem}
There is a trade-off between the time-steps $\Delta t$ and the time it takes to find a point in which $G(t) < \frac{M}{2}\Delta t + \varepsilon$, which is $t = \frac{S(x(0),\eta(0))}{\varepsilon}$. On one hand, we want larger time-steps (to avoid high-rate sampling) and shorter overall times, however, increasing both $\Delta t$ and $\varepsilon$ creates a worse eventual bound on $G(t)$. We can choose both by maximizing an appropriate cost function $C(\Delta t,\varepsilon)$, monotone in both $\Delta t$ and $\varepsilon$, subject to $\frac{M}{2}\Delta t + \varepsilon = \delta_1, \varepsilon \ge 0, \Delta t \ge 0$. Choosing $C(\Delta t,\varepsilon)$ as linear is inadvisable, as the maximizing a linear function with linear constraints always leads to the optimizer being on the boundary, which means either $\Delta t=0$ or $\varepsilon = 0$. The choice $\Delta t = \frac{\delta_1}{M}$ and $\varepsilon = \frac{\delta_1}{2}$ mentioned above corresponds to the geometric average cost function $C(\Delta t,\varepsilon) = \sqrt{\Delta t \varepsilon}$. Other choices of $C$ can express practical constraints, e.g. relative apathy to large convergence times relative to high-rate sampling should result in a cost function penalizing small values of $\Delta t$ more harshly than small values of $\varepsilon$.
\end{rem}

\subsection {Asserting Convergence Using Convergence Profiles} \label{subsec.ConvProf}
For this subsection, we now assume that Assumption \ref{Assumption_g_2} holds and that the agents are output-strictly MEIP, i.e., that $\rho_i > 0$.
Consider \eqref{eq.PassivityInequality1} and suppose there is a non-negative monotone function $\mathcal F$ such that for any $t$, the right-hand side of \eqref{eq.PassivityInequality1} is bounded from above by $-\mathcal F(S)$. In that case, we get an estimate of the form $\dot{S} \le -\mathcal F(S)$. This is a weaker estimate than \eqref{eq.PassivityInequality1}, but it has a more appealing discrete-time form,
\begin{align} \label{eq.AutonomousODEOnS}
S(x(t_{k+1})) - S(x(t_{k})) \le& -\int_{t_{k}}^{t_{k+1}}\mathcal F(S(x(t)))dt \\\le& - \mathcal F(S(x(t_{k+1})))\cdot(t_{k+1}-t_k), \nonumber
\end{align}
where we use the monotonicity of $\mathcal F$ and the fact that $S(x(t))$ is monotone non-ascending. Due to Assumption \ref{Assumption_g_2}, we focus on the elements of the right-hand side of \eqref{eq.PassivityInequality1} corresponding to the agents, and neglect the ones corresponding to controllers, as $S$ is now the sum of the functions $S_i(x_i)$. Because controllers are passive, we have $\nu_e \ge 0$, so removing the said term does not change the inequality's validity.

In order to find $\mathcal F$, it is natural to look for functions $\Omega_i$ satisfying $\Omega_i(S_i) \le (y_i(t)-\mathrm y_i)^2$.  We define the existence of the functions $\Omega_i$ properly in the following definition.
\begin{defn}
Let $\Omega:[0,\infty)\to[0,\infty)$ be any function on the non-negative real numbers. We say that an autonomous system has the \emph{convergence profile} $(\rho,\Omega)$ with respect to the steady-state $(\mathrm {u,y})$ if there exists a $C^1$ storage function $S(x)$ such that the following inequalities hold:
\begin{itemize}
	\item[i)] $\frac{dS(x(t))}{dt} \le (u(t)-\mathrm u)(y(t)-\mathrm y) - \rho(y(t)-\mathrm y)^2 $,
	\item[ii)] $\Omega(S(x(t))) \le (y(t) - \mathrm y)^2. $
\end{itemize}
\end{defn}

\begin{exam}
Consider the SISO system $\Sigma$ defined by $\dot{x} = -x + u,\ y = x$, and consider the steady-state input-output pair $(0,0)$. The storage function $S(x(t)) = \frac{1}{2}x(t)^2$ satisfies
\begin{align*}
\dot{S}(x(t)) =& x(t)\dot{x}(t) = (u(t)-0)(y(t) -0)- (y(t)-0)^2 .
\end{align*}
Thus $\Sigma$ has convergence profile $(1,\Omega)$ for $\Omega(\theta) = 2\theta$. 
\end{exam}

More generally, when considering an LTI system with no input-feedthrough, both functions $S(x)$ and $(y(t)-\mathrm y)^2$ are quadratic in $x$. Thus there is a monotone linear function $\Omega$ such that the inequality $\Omega(S(x(t)) \le (y(t)-\mathrm y)^2$ holds. In particular, a function $\Omega$ exists in this case. We show that functions $\Omega$ exist for rather general systems.

\begin{thm} \label{thm.OmegaExists}
Let $\Sigma$ be the SISO system of the form {$\dot{x} = -f(x)+q(x)u,\, y=h(x)$}. Suppose $q$ is a positive continuous function, that $f/q$ is $C^1$ and monotone ascending and that $h$ is $C^1$ and strictly monotone ascending. Let $(\mathrm u=f(\mathrm x)/q(\mathrm x),y = h(\mathrm x))$ be any steady-state input-output pair of the system. Then
\begin{enumerate}
\item[i)] using the storage function $S(x) = \int_{\mathrm x} ^ x \frac{h(\sigma)-h(\mathrm x)}{q(\sigma)}d\sigma$, the system $\Sigma$ has the convergence profile $(\rho,\Omega)$ for a strictly ascending function $\Omega$ and $\rho = \inf_{x} \frac{f(x)-f(\mathrm x)}{h(x)-h(\mathrm x)}\ge 0$;
\item[ii)] \label{Part.PowerLaw}  suppose there exists some $\alpha > 0$ such that the limit $\lim_{x \to \mathrm x} \frac{|h(x)-h(\mathrm x)|}{|x-\mathrm x|^\alpha}$ exists and is finite. Then the limit $\lim_{\theta\to 0} \frac{\Omega(\theta)}{\theta^\beta}$ exists and is finite, where $\beta = \frac{2\alpha}{\alpha+1}$. {In other words, if $h$ behaves like a power law near $\mathrm x$, then $\Omega$ behaves like a power law near $0$}.  
\end{enumerate}
\end{thm}

\begin{proof}
The passivation inequality follows from Theorem \ref{thm.MEIPAgents}, so we focus on constructing the function $\Omega$. For every $\theta\ge 0$, we define the set $\mathbb A_\theta = \{x\in \mathbb{R}:\ (h(x)-h(\mathrm x))^2 \le \theta\}$. We want that $x \in \mathbb{A}_\theta$ would imply that that $\Omega(S(x))\le \theta$. Because $h$ is continuous and monotone, it is clear that $\mathbb{A}_\theta$ is an interval containing $\mathrm x$. Now, let $\omega$ be the function on $[0,\infty)$ defined as $\omega(\theta) = \sup_{x\in\mathbb{A}_\theta} S(x)$. We note that $\omega$ can take infinite values (e.g. when $h$ is bounded, but $S$ is not). However, we show that the restriction of $\omega$ on $\{\theta:\ \omega(\theta) < \infty\}$ is strictly monotone. If we show that this claim is true, then $\omega$ has an inverse function {which is also strictly monotone}. Define $\Omega = \omega^{-1}$ as the strictly monotone inverse function. By definition, for any $x\in\mathbb{R}$ we have that $x \in \mathbb{A}_\theta$ for $\theta = (h(x)-h(\mathrm x))^2$, so $S(x) \le \omega(\theta)$. Thus 
$
\Omega(S(x)) \le \Omega(\omega(\theta)) = \theta = (h(x) - h(\mathrm x))^2,
$
concluding the first part of the proof. 

We now prove that the restriction of $\omega$ on $\{\theta:\ \omega(\theta) < \infty\}$ is strictly monotone. It is clear that if $0\le \theta_1 < \theta_2$ then the interval $\{x:\ (h(x)-h(\mathrm x))^2 \le \theta_1\} = \mathbb{A}_{\theta_1}$ is strictly contained in the interval $\{x:\ (h(x)-h(\mathrm x))^2 \le \theta_2\} = \mathbb{A}_{\theta_2}${, as $h$ is strictly monotone. Moreover, It is clear that $S$ is strictly ascending in $\left[\mathrm x,\infty\right)$ and strictly descending in $\left(-\infty,\mathrm x\right]$, as the function $\frac{h(x)-h(\mathrm x)}{g(x)}$ is positive on $(\mathrm x,\infty)$ and negative on $(-\infty,\mathrm x)$}. Thus we have $\omega(\theta_1) < \omega(\theta_2)${, unless $\omega(\theta_1) = \infty$}, which is what we wanted to prove.

We now move to the second part of theorem, in which we show that if $h$ behaves like a power law near $\mathrm x$, then $\Omega$ behaves like a power law near zero. We use big-$O$ notation (in the limit $x\to \mathrm x$). By assumption and strict monotonicity of $h$, we have:
\begin{align} \label{eq.hPowerLaw}
h(x)-h(\mathrm x) = C\mathrm{sgn}(x-\mathrm x)|x-\mathrm x|^\alpha + o(|x-\mathrm x|^\alpha)
\end{align}
for some constant $C>0$, implying
$
(h(x)-h(\mathrm x))^2 = C^2 |x-\mathrm x|^{2\alpha} + o(|x-\mathrm x|^{2\alpha}).
$
By definition, we conclude that for $\theta>0$ small enough, $\mathbb A_\theta$ is an interval centered at $\mathrm x$ and has radius $\theta^{1/2\alpha}/C^{1/\alpha} + o(\theta^{1/2\alpha})$.
We recall that $S(x) = \int_\mathrm x^x \frac{h(\sigma)-h(\mathrm x)}{q(\sigma)}d\sigma$. We write $q(x) = q(\mathrm x) + o(1)$  as $q$ is continuous, so \eqref{eq.hPowerLaw} implies that  $$\frac{h(\sigma)-h(\mathrm x)}{q(\sigma)} = \frac{1}{q(\mathrm x)}\left(C\mathrm{sgn}(x-\mathrm x)|x-\mathrm x|^\alpha + o(|x-\mathrm x|^\alpha)\right).$$ Thus, $S(x) = \frac{C}{q(\mathrm x)(\alpha+1)}|x-\mathrm x|^{\alpha+1} + o(|x-\mathrm x|^{\alpha+1})$. We now compute $\omega(\theta)$ by definition, using our characterization of $\mathbb A_\theta$:{\small
\begin{align*}
&\omega(\theta) = \max_{x\in \mathbb A_\theta} S(x)
 = \max_{x\in \mathbb A_\theta} \bigg(\frac{C|x-\mathrm x|^{\alpha+1}}{q(\mathrm x)(\alpha+1)} + o(|x-\mathrm x|^{\alpha+1})\bigg) \\&
=\frac{C}{q(\mathrm x)(\alpha+1)} \bigg(\frac{\theta^{\frac{1}{2\alpha}}}{C^{1/\alpha}}\bigg)^{\alpha+1} + o\left((\theta^{\frac{1}{2\alpha}})^{\alpha+1}\right) = (D+o(1))\theta^{\frac{\alpha+1}{2\alpha}}
\end{align*}\normalsize}
for $D=\frac{1}{q(\mathrm x)(\alpha+1)C^{1/\alpha}}>0$. Thus, the inverse function $\Omega(\theta)$ is given as $\Omega(\theta) = (D^{-\frac{2\alpha}{1+\alpha}}-o(1))\theta^{\frac{2\alpha}{1+\alpha}}$, as desired.
\end{proof}

\begin{exam}
Consider a system with $q(x) = 1, h(x) = \sqrt[3]{x}$ and a steady state $\mathrm u = \mathrm x = \mathrm y = 0$. $h(x)$ behaves like a power law with power $\alpha = \frac{1}{3}$. Part ii) of Theorem \ref{thm.OmegaExists} implies that $\Omega$ also behaves like a power law, with power $\beta = \frac{2\alpha}{\alpha+1} = \frac{1}{2}$. We exemplify the computation of $\Omega$ as done in the proof, and show it behaves like a power law with $\beta = \frac{1}{2}$, as forecasted by the theorem.
Indeed, $S(x) = \int_0^x \sqrt[3]{\sigma}d\sigma = \frac{3}{4}x^{4/3}$, and $(h(x)-h(\mathrm x))^2 = x^{2/3}$. 

For every $\theta \ge 0$, $\mathbb{A}_\theta = \{x:\ x^{2/3} \le \theta\} = [-\theta^{1.5},\theta^{1.5}]$. Thus
$
\omega(\theta) = \sup_{x\in\mathbb{A}_\theta} S(x) = \sup_{|x|\le \theta^{1.5}} \frac{3}{4}x^{4/3} = \frac{3}{4}\theta^2,
$
implying that $\Omega$, the inverse function of $\omega$, is given by $\sqrt{\frac{4}{3}\theta}$, and one observes that actually $(h(x)-h(\mathrm x))^2 = \Omega(S(x))$. 
\end{exam}

\begin{rem}
Theorem \ref{thm.OmegaExists} gives a prescription to design the function $\Omega$. However, some steps, namely the inversion of $\omega$, are computationally hard. For example, if $h(x) = 1-e^{-x}$ and $q(x) = 1$, then $\omega(\theta) = \log_{e}{\frac{1}{1-\sqrt{\theta}}} - \sqrt{\theta}$ for $\theta < 1$ and $\omega(\theta) = \infty$ for $\theta \ge 1$, which is almost impossible to invert analytically. To solve this problem, we can either precompute the different values of $\Omega$ numerically and store them in a table, or approximate them on-line using the bisection method. The strength of Theorem \ref{thm.OmegaExists} is that it shows that a function $\Omega$ can always be found, implying this method is always applicable. 
\end{rem}

Up until now, we managed to transform the equation \eqref{eq.PassivityInequality} to the equation $
\frac{dS}{dt} \le \sum_i -\rho_i \Omega_i(S_i)$,
for some non-negative monotone functions $\Omega_i$. This is closer to an inequality of the form $\dot{S} \le -\mathcal F(S)$, but we still cannot use it without high-rate sampling, as we cannot assume that $S_i(x_i(t))$ are monotone decreasing. We want to transform the right hand side into a function of $S$. We note that $\Omega_i(\theta_i) = 0$ only at $\theta_i = 0$, as $S_i=0$ happens only at $\mathrm x_i$. We claim the following:
\begin{prop} \label{prop.Jensen}
Let $\rho_1,...,\rho_n$ be any positive numbers, and let $\Omega_1,...,\Omega_n:[0,\infty)\to [0,\infty)$ be any $C^1$ strictly monotone functions such that $\Omega_i(\theta_i) = 0$ only at $\theta_i = 0$. Suppose further that for any $i$ there exists some $\beta_i>0$ such that the limit $\lim_{\theta_i\to 0} \frac{\Omega_i(\theta_i)}{\theta_i^{\beta_i}}$ exists and is positive. Define $\Omega_\star:[0,\infty)\to [0,\infty)$ as $\Omega_\star(\theta) = \min_i \Omega_i(\theta)$. Then for every $D>0$, there exists some constant $C>0$ such that for all $D\ge\theta_1,\cdots,\theta_n\ge 0$, we have $\sum_{i=1}^n \rho_i \Omega_i(\theta_i) \ge C\cdot\Omega_\star(\sum_{i=1}^n \theta_i)$.
\end{prop}
{The proof of the proposition can be found in the appendix.}
\begin{cor} \label{cor.FFunction}
Let $S_1,...,S_n$ be the storage functions of the agents, let $S=\sum_i S_i$, and let $\Omega_1,\cdots,\Omega_n$ be $C^1$ strictly monotone functions such that $\Omega_i(\theta_i) = 0$ only at $\theta_i = 0$. Suppose that for any $i$ there exists some $\beta_i>0$ such that the limit $\lim_{\theta_i\to 0} \frac{\Omega_i(\theta_i)}{\theta_i^{\beta_i}}$ exists and is positive. Moreover, Suppose that $\dot{S} \le \sum_i \rho_i \Omega_i(S_i)$. Then for every bounded set $B\subset \mathbb{R}^n$ there exists a constant $C>0$ such that for any trajectory of the closed-loop system with initial condition in $B$, the inequality $\dot{S} \le -C\cdot \Omega_\star(S)$ holds, where $\Omega_\star(\theta) = \min_i \Omega_i(\theta)$.
\end{cor}
\begin{proof}
Use $\theta_i = S_i$  and $D = S(x(0))$ in Proposition \ref{prop.Jensen}. 
\end{proof}

Proposition \ref{prop.Jensen} and Corollary \ref{cor.FFunction} show that an inequality of the form \eqref{eq.AutonomousODEOnS} can be achieved, so long the functions $\Omega_i$ from Theorem \ref{thm.OmegaExists} ``behave nicely" around $0$, namely do not grow faster nor slower than a power law. This condition is very general, and only excludes pathologies as $\Omega(\theta) = \frac{1}{\log(1/\theta)}$, growing faster than any power law, and $\Omega(\theta) = \exp(-1/\theta^2)$, growing slower than any power law. Theorem \ref{thm.OmegaExists} shows that if $h$ behaves like a power law near $\mathrm x$, then so does $\Omega$, so pathological functions $\Omega_\star$ can only come from pathological measurement functions $h_i$. We show it is enough to check the discretized equation \eqref{eq.AutonomousODEOnS} to assert convergence.

\begin{prop} \label{prop.ConvRateInequality}
Let $\Omega_\star:[0,\infty)\to [0,\infty)$ be any continuous function such that $\Omega_\star(\theta) = 0$ only at $\theta = 0$. Let $\tilde S(t)$ be any time-dependent monotone decreasing function $\tilde S:[0,\infty)\to [0,\infty)$. Let $t_1,t_2,t_3,\cdots$ be any unbounded sequence of times such that $\mathrm{liminf}_{k\to\infty} (t_{k+1}-t_k) > 0$, and suppose that for every $k$, the inequality $\tilde S(t_{k+1}) - \tilde S(t_k) \le -\Omega_\star(\tilde S(t_{k+1}))(t_{k+1} - t_k)$ holds. Then $\tilde S(t)\to 0$ as $t\to \infty$.
\end{prop}
{The proof of the proposition can be found in the appendix.}{
We want to use $\tilde S(t) = S(x(t))$. The results above suggest an algorithm for convergence assertion.
\begin{algorithm} [h]
\caption{Convergence Assertion using Convergence Profile}
\label{alg.ConvProfile}
\textbf{Input:} A diffusively-coupled network $(\G,\Sigma,\Pi)$, an initial condition $x(0)$ and a conjectured steady-state $\hat{\mathrm x}$.
\begin{algorithmic}[1]
\State Let $S_i(x_i) = \int_{\hat{\mathrm x}_i}^{x_i} \frac{h_i(\sigma_i) - h_i(\hat{\mathrm x}_i)}{q_i(\sigma_i)}d\sigma_i$, $S(x) = \sum_i S_i(x_i)$.
\State Use Theorem \ref{thm.OmegaExists}, Proposition \ref{prop.Jensen} and Corollary \ref{cor.FFunction} to find a function $\Omega$ such that $\dot{S} \le -\Omega(S)$ for all times $t$, with initial condition $x(0)$. \;
\State Choose $\delta_0 = S(x(0))$ and $t_0 = 0$\;
\For {$k=0,1,2,3,\cdots$}
\State Define $\delta_{k+1} = \delta_{k}/2$.;
\State Define $M = \min_{x:\ S(x)\ge \delta_{k+1}} \Omega(S(x))$\;
\State Sample the system at time some $t_{k+1} > t_{k} + \frac{S(x_0)}{M}$.
\If {\small$S(x(t_{k+1}))-S(x(t_k)) \not\le -\Omega(S(x(t_{k+1})))(t_{k+1}-t_k)$\normalsize}
\State Stop and return ``no";
\EndIf
\EndFor
\end{algorithmic}
\end{algorithm}
\begin{thm}
 Algorithm \ref{alg.ConvProfile}, taking the system $(\G,\Sigma,\Pi)$, the initial state $x(0)$, and the conjectured steady-state $\hat{\mathrm{x}} = h^{-1}(\mathrm y)$ as input, satisfies Assumption \ref{assump.Diff}.
\end{thm}
\begin{proof}
We denote the true limit of the system $(\G,\Sigma,\Pi)$ by $\mathrm x$. We first assume the algorithm never stops, and show that $\hat{\mathrm x} = \mathrm x$ . We show that $S(x(t_k)) \le \delta_k$, which would suffice as $\delta_k\to 0$ and $S(t)\to 0$ implies that $x(t) \to \hat{\mathrm x}$, and thus $\mathrm x = \hat{\mathrm x}$. 
Suppose, heading toward contradiction, that $S(x(t_k)) \not \le \delta_k$. Then $\Omega(S(x(t_k))) \ge M$, meaning that the right-hand side of the checked inequality is larger than $-S(x(t_k))$. Thus, if the inequality holds then $S(x(t_{k+1})) < 0$, which is absurd. Thus $S(x(t_k)) \le \delta_k$, and $\hat{\mathrm x} = \mathrm x$.
On the contrary, if the conjectured limit $\hat{\mathrm x}$ is the true limit of the network, then Theorem \ref{thm.OmegaExists}, Proposition \ref{prop.Jensen} and Corollary \ref{cor.FFunction} show that $S(x(t_{k+1}))-S(x(t_k)) \le -\Omega(S(x(t_{k+1})))(x(t_{k+1})-x(t_k))$ always holds, so the algorithm never stops, as expected.
\end{proof}
}

\begin{rem} \label{rem.ConvRateAlgorithm}
{Proposition \ref{prop.ConvRateInequality} shows we can take any sample times $t_k$ such that $\lim\inf_{k\to\infty} (t_k - t_{k-1}) > 0$, and still get a valid convergence assertion algorithm. The suggested algorithm gives extra information, as it also bounds the distance of $x(t_k)$ from $\mathrm x$.}
Another way to choose $t_k$ is to use the solution of the ODE $\dot{\tilde{S}} = -\Omega(\tilde{S})$ with $\tilde{S}(t_0) = S(x_0)$. Let $t_k$ be the earliest time in which $\tilde{S}(t_k) \le \delta_k$. The inequality $\dot{S} \le -\Omega(S)$ assures that $S(x(t_k)) \le \delta_k$. This method is more demanding, as the minima $M$ computed before can be stored in a table, but the solution to the ODE must be computed on-line.
\end{rem}

\begin{rem}
Although we can prove convergence with this method using very seldom measurements, we should still sample the system at a reasonable rate. This is because we want to detect faults as soon as possible. If we sample the system in too large intervals, we will not be able to sense a fault until a large amount of time has passed.
\end{rem}

We conclude this section with a short discussion about the perks and drawbacks of the two presented convergence assertion methods. The convergence profile method allows the designer to sample the system at any desired rate, allowing one to prove convergence using very seldom measurements. Moreover, it gives certain rate of convergence guarantees before running the system. On the contrary, the high-rate sampling method can require a long time to assert convergence to a $\delta$-ball around the desired steady-state, unless one is willing to increase the sampling rate, perhaps arbitrarily. However, its main upshot over the convergence profile method is that we need not assume that Assumption \ref{Assumption_g_2} holds{, and that the method is computationally easier, as one can avoid function inversion needed to compute $\Omega$}.
\vspace{-5pt}
\section{Case Study: Velocity Coordination in Vehicles with Drag} \label{sec.Simulation}
{Consider a collection of $n=20$ vehicles trying to coordinate their velocity. Each vehicle is modeled as a double integrator $G(s)=1/s^2$ in vacuum, but is subject to drag in the real world. The drag model of a vehicle is usually modeled as a force against the direction of the vehicle's motion, which is quadratic in the size of the velocity \cite[p. 164]{Serway2018}. Thus, each vehicle has the following model:
\begin{align}
\Sigma_i: \dot{x}_i = -C_i x_i|x_i| + u_i,~~ y_i &= x_i,
\end{align}
where $x_i$ is the velocity of the $i$-th vehicle, and $C_i$ is a constant which aggregates different parameters affecting the drag, e.g. density of the air, viscosity of the air, the geometry of the vehicle, and the mass of the vehicle. The vehicles are trying to coordinate their velocity - Agents $\#1-\#7$ want to travel at $60_{\rm km/h}$, agents $\#8-\#13$ want to travel at $70_{\rm km/h}$, and agents $\#14-\#20$ want to travel at $50_{\rm km/h}$. The edge controllers are static nonlinearities given by sigmoid functions of the form $\mu_e = \tanh(\zeta_e)$. This diffusively coupled network satisfies both Assumptions \ref{Assumption} and \ref{Assumption_g_2}, and we note both the agents and the controllers are nonlinear. We choose an interaction graph $\G$ as seen in Figure \ref{fig.VehicleExampleGraph}. One can check that $\G$ is $4$-connected, either using Menger's theorem or using other known algorithms \cite{Galil1980}. }

\begin{figure} [!t] 
\vspace{-15pt}
    \centering
    \includegraphics[scale=0.5]{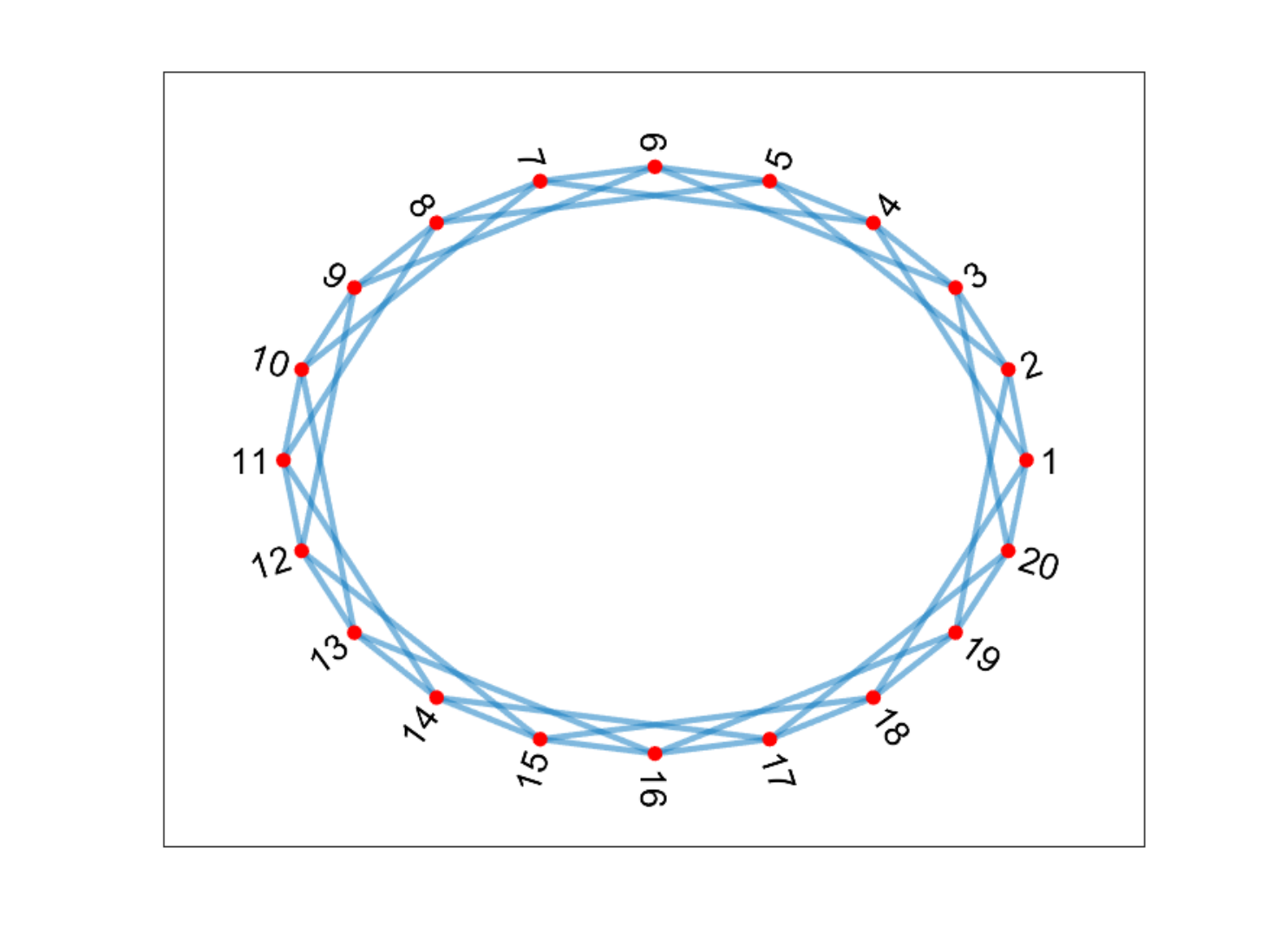}
\vspace{-20pt}
    \caption{The faultless underlying graph in the case study.}
    \label{fig.VehicleExampleGraph}
\vspace{-10pt}
\end{figure}
{
The parameters $C_i$ were chosen as log-uniformly between $0.01$ and $0.1$. The initial velocity of the agents was chosen to be Gaussian with mean $\mu = 70_{km/h}$ and standard deviation $\sigma = 20_{km/h}$.
We solve the synthesis problem, forcing the network to converge to $\mathrm y^\star = [60\cdot\mathbbm{1}_{7}^\top,70\cdot\mathbbm{1}_6^\top,50\cdot\mathbbm{1}_7^\top]^\top_{\rm km/h}$, where $\mathbbm{1}_m$ is the all-one vector of size $m$, allowing up to $2$ edges to fault. We run our FDI protocol, implementing the profile-based convergence assertion scheme, sampling the system at $10_{\rm Hz}$ (i.e., a modified version of Algorithm \ref{alg.ConvProfile}). We consider four different scenarios., each lasting $100$ seconds.
\begin{enumerate}
\item A faultless scenario.
\item At time $t=20_{sec}$, the edge $\{2,3\}$ faults, and at time $t=50_{sec}$, the edge $\{13,14\}$ faults.
\item At time $t=20_{sec}$, the edge $\{2,3\}$ faults, and at time $t=21_{sec}$, the edge $\{13,14\}$ faults.
\item At time $t=0.5_{sec}$, the edge $\{2,3\}$ faults, and at time $t=4_{sec}$, the edge $\{13,14\}$ faults.
\end{enumerate} 
The first scenario tests the nominal behavior of the protocol. The second tests its ability to handle with one single fault at a time. The third tests its ability to handle more than one fault at a time. The last tests its ability to deal with faults before the network converged.
The results  are available in Figures \ref{fig.VehicleExampleNominal}, \ref{fig.VehicleExampleFarFaults}, \ref{fig.VehicleExampleCloseFaults}, and \ref{fig.VehicleExampleCloseEarlyFaults}. It can be seen that we achieve the control goal in all four scenarios. Moreover, in all scenarios and at all times, the velocities of the agents are not too far from the values found in $y^\star$, meaning that this protocol cannot harm the agents by demanding them to have very wild states. 
In the second and third scenario, the exploratory phases begins at the first measurement after the fault occurred. On the contrary, in the fourth scenario, it takes the exploratory phase begins only at $t=1.3_{sec}$, a little under a second after the first fault. This is because the steady-states of the faulty and nominal closed-loop system are relatively close, meaning it takes a little extra time to find that a fault exists. }
\begin{figure} [!t] 
    \centering
    \includegraphics[scale=0.55]{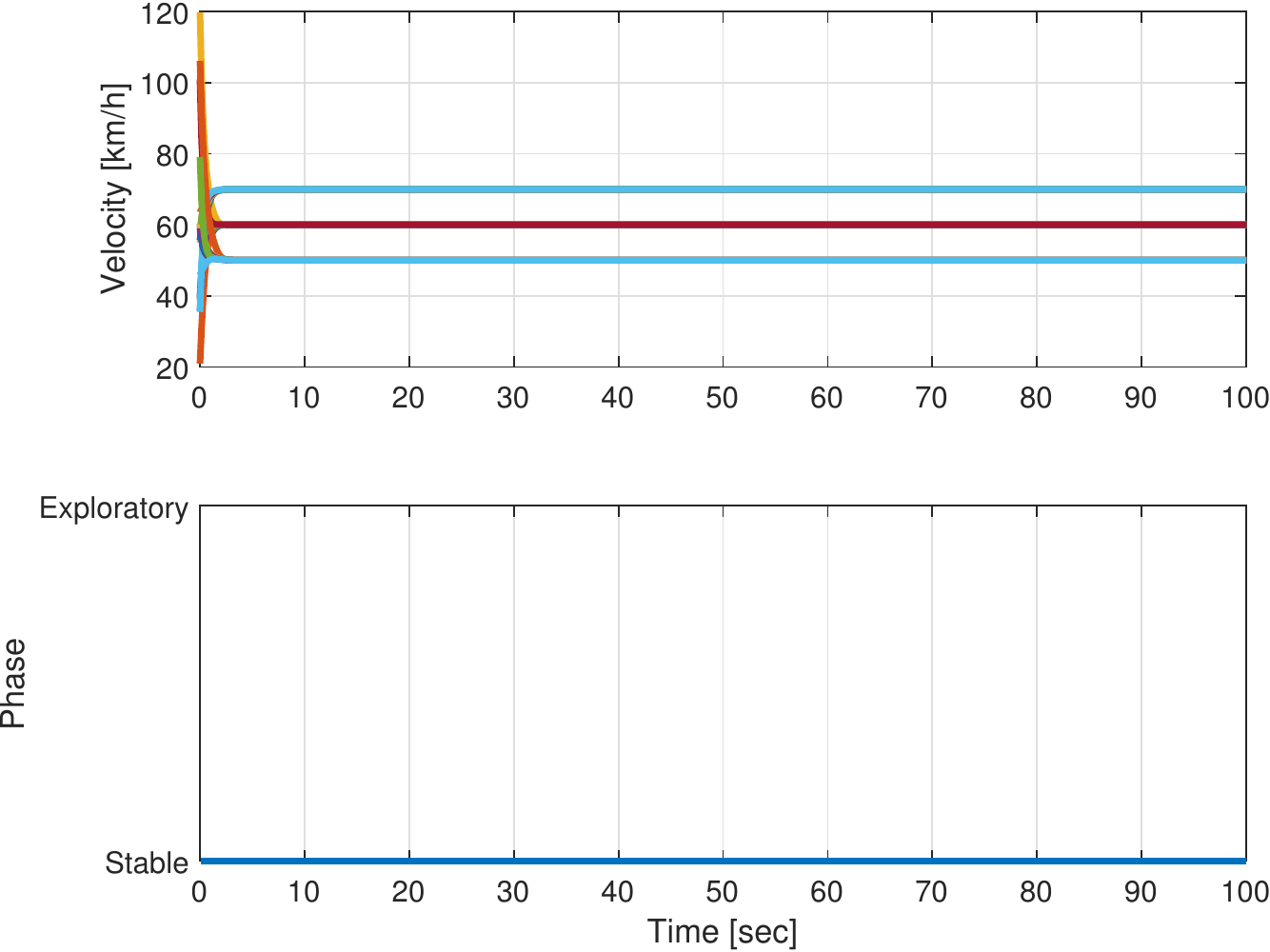}
    \caption{Results of first scenario.}
    \label{fig.VehicleExampleNominal}
\vspace{-15pt}
\end{figure}
\begin{figure} [!b] 
\vspace{-15pt}
    \centering
    \includegraphics[scale=0.55]{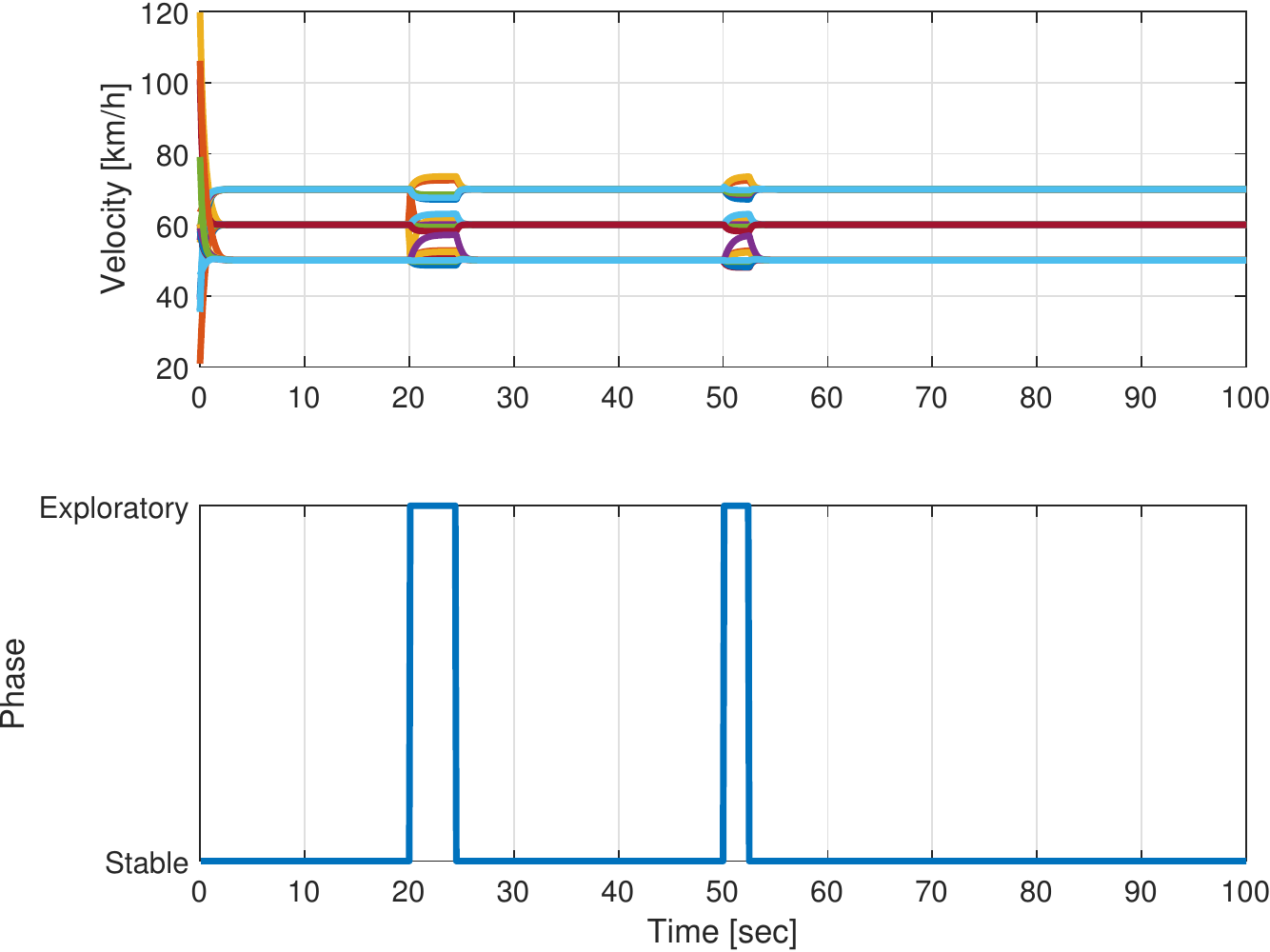}
    \caption{Results of second scenario.}
    \label{fig.VehicleExampleFarFaults}
\end{figure}
\begin{figure} [!t] 
    \centering
    \includegraphics[scale=0.55]{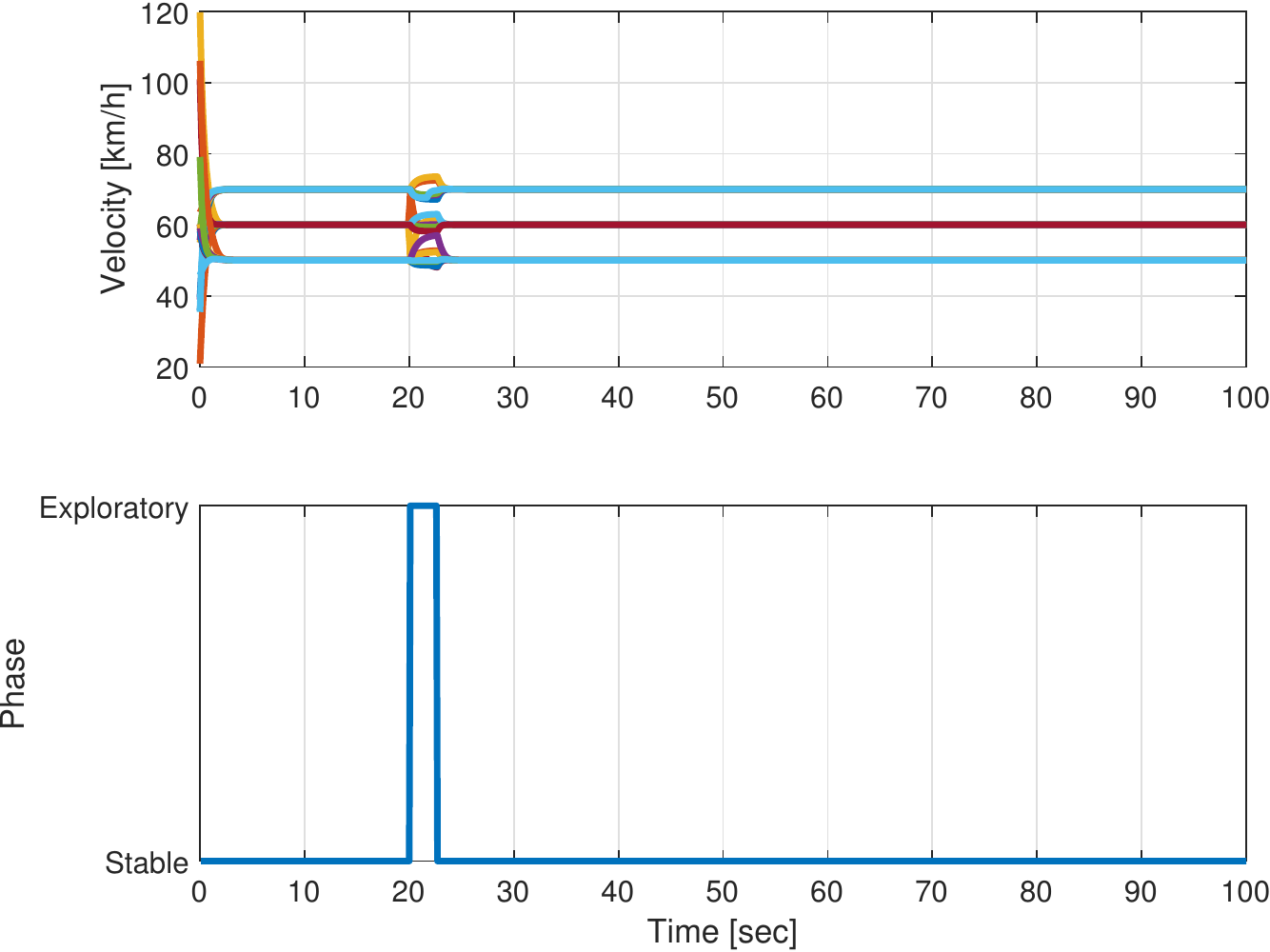}
    \caption{Results of third scenario.}
    \label{fig.VehicleExampleCloseFaults}
\vspace{-15pt}
\end{figure}
\begin{figure} [!b] 
\vspace{-15pt}
    \centering
    \includegraphics[scale=0.55]{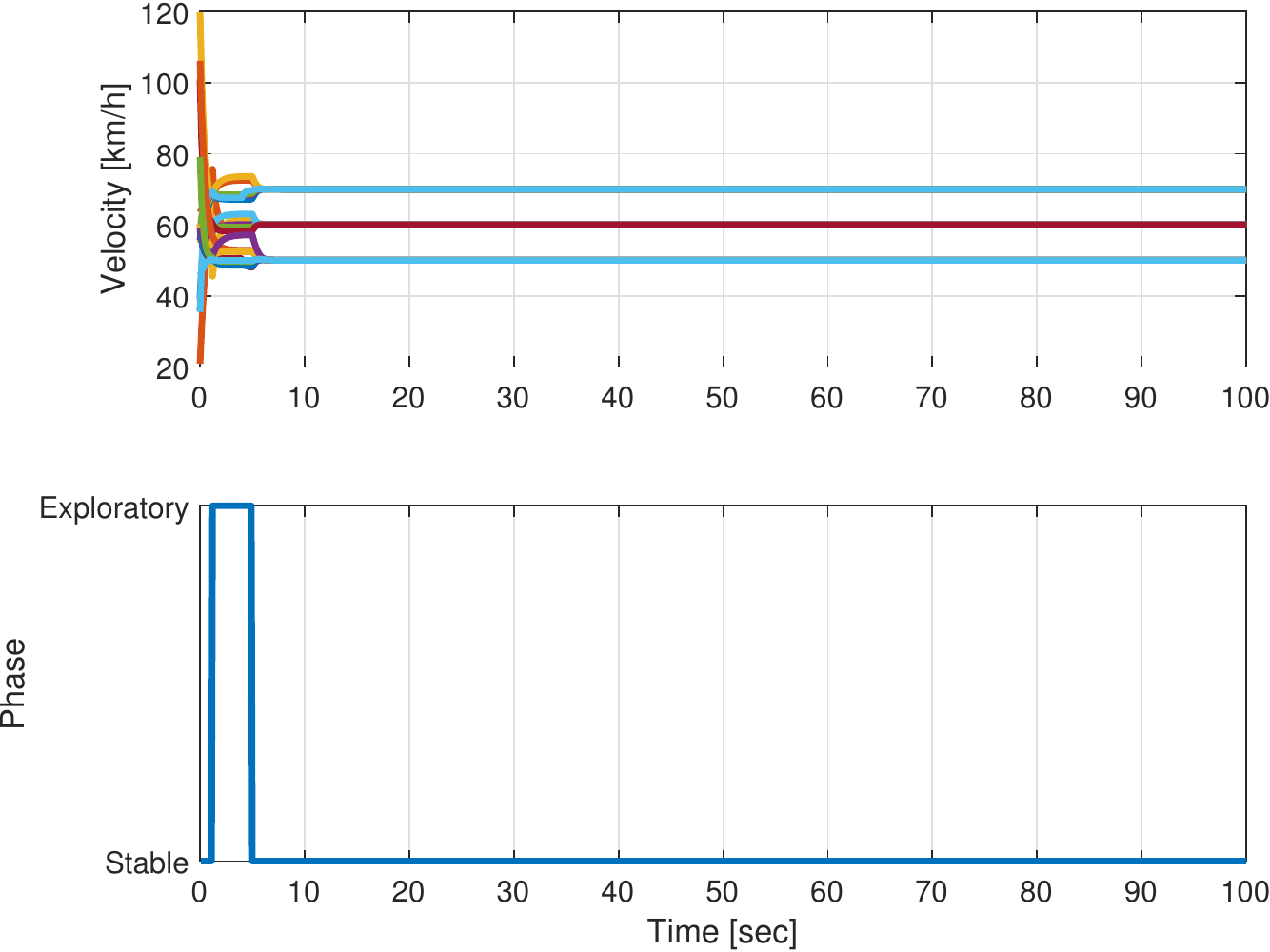}
    \caption{Results of fourth scenario.}
    \label{fig.VehicleExampleCloseEarlyFaults}
\end{figure}

\if(0)
\begin{figure}[!t]
\begin{center}
	\subfigure[Results of first scenario] {\scalebox{.5}{\includegraphics{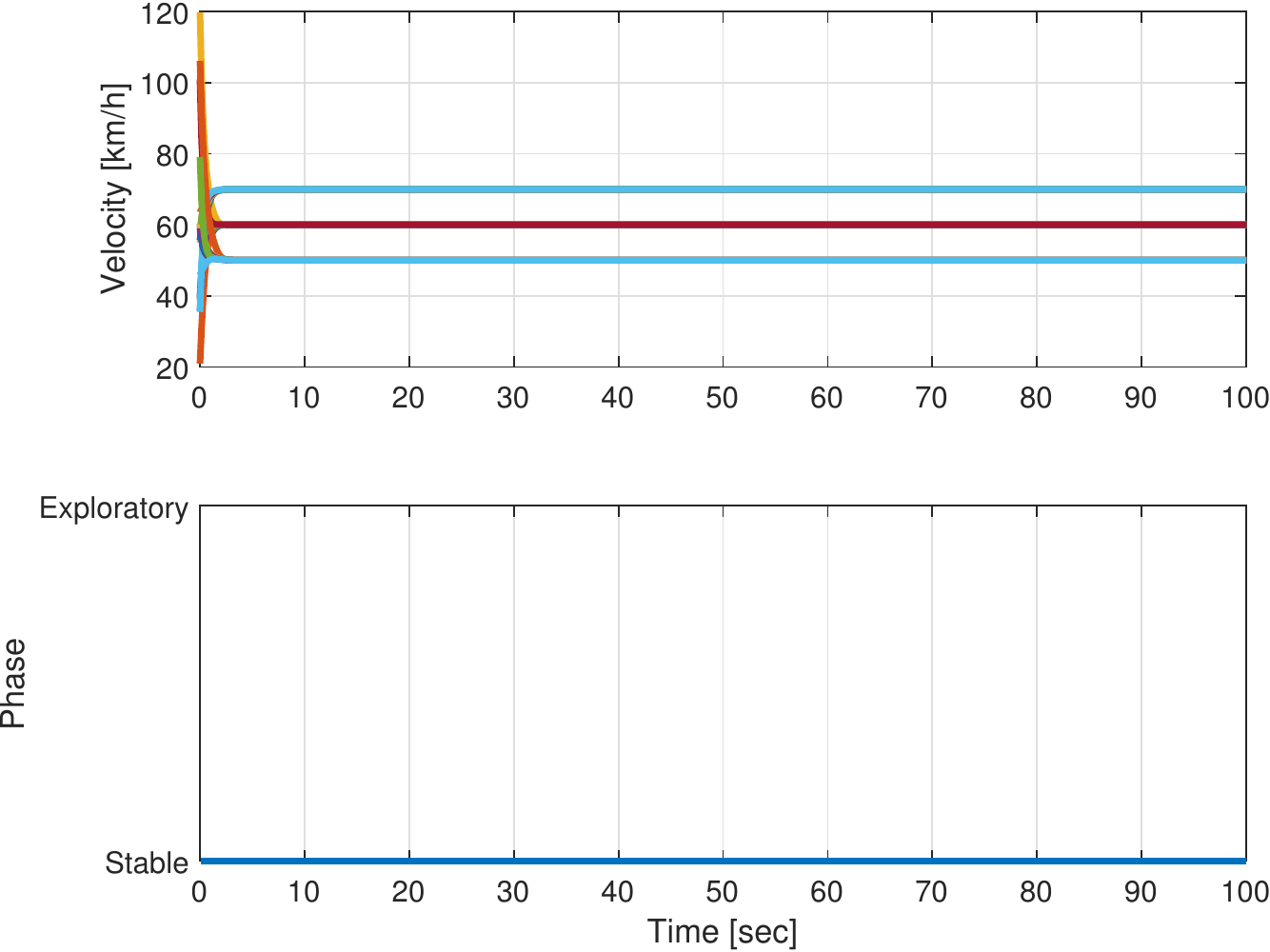}}\label{fig.VehicleExampleNominal}}
	\subfigure[Results of second scenario]{\scalebox{.5}{\includegraphics{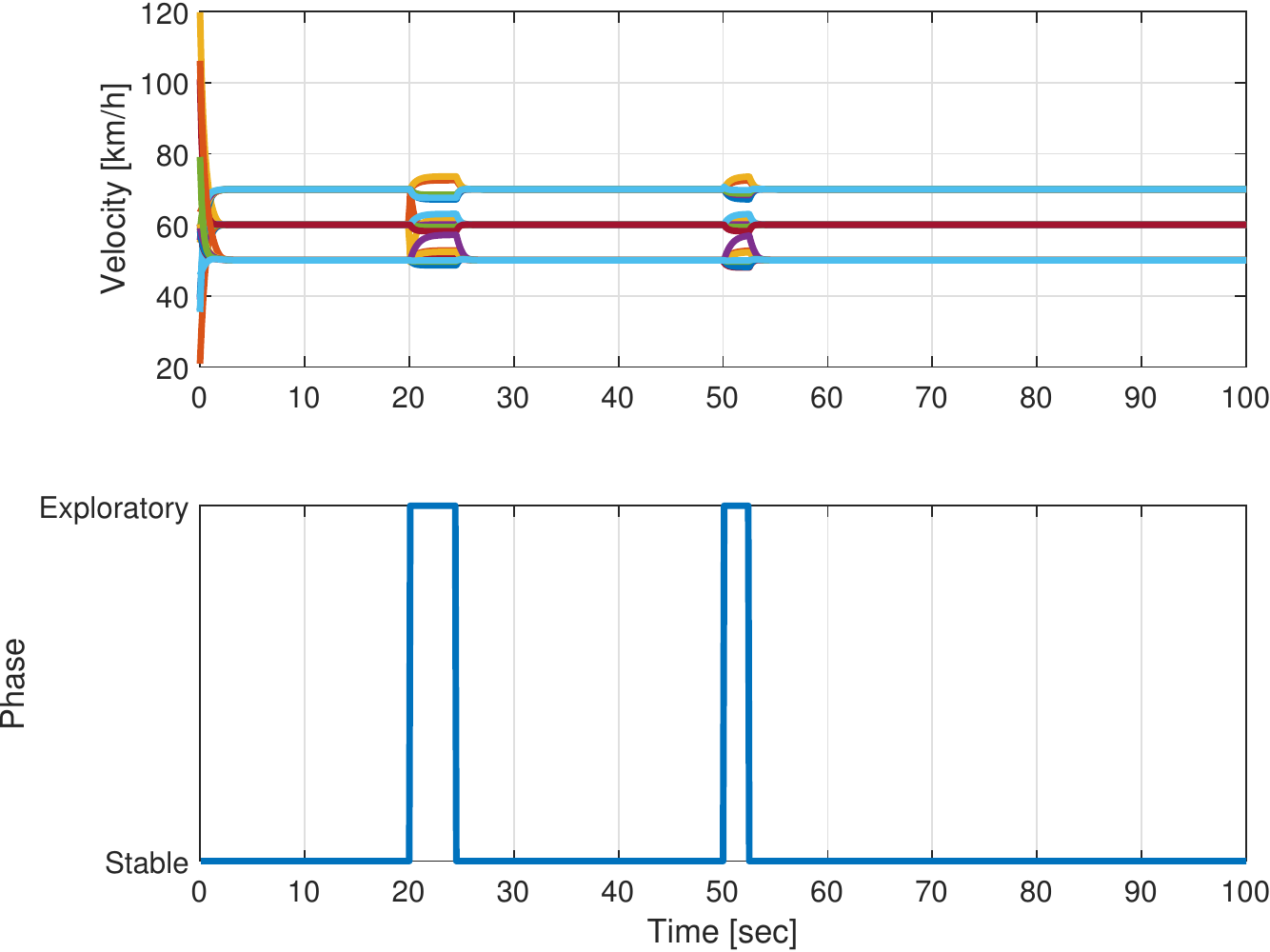}}\label{fig.VehicleExampleFarFaults}}
  \caption{First scenario in FDI for vehicle systems.}
\end{center}
\end{figure}
\begin{figure}[!b]
\begin{center}
	\subfigure[Results of Third Scenario] {\scalebox{.5}{\includegraphics{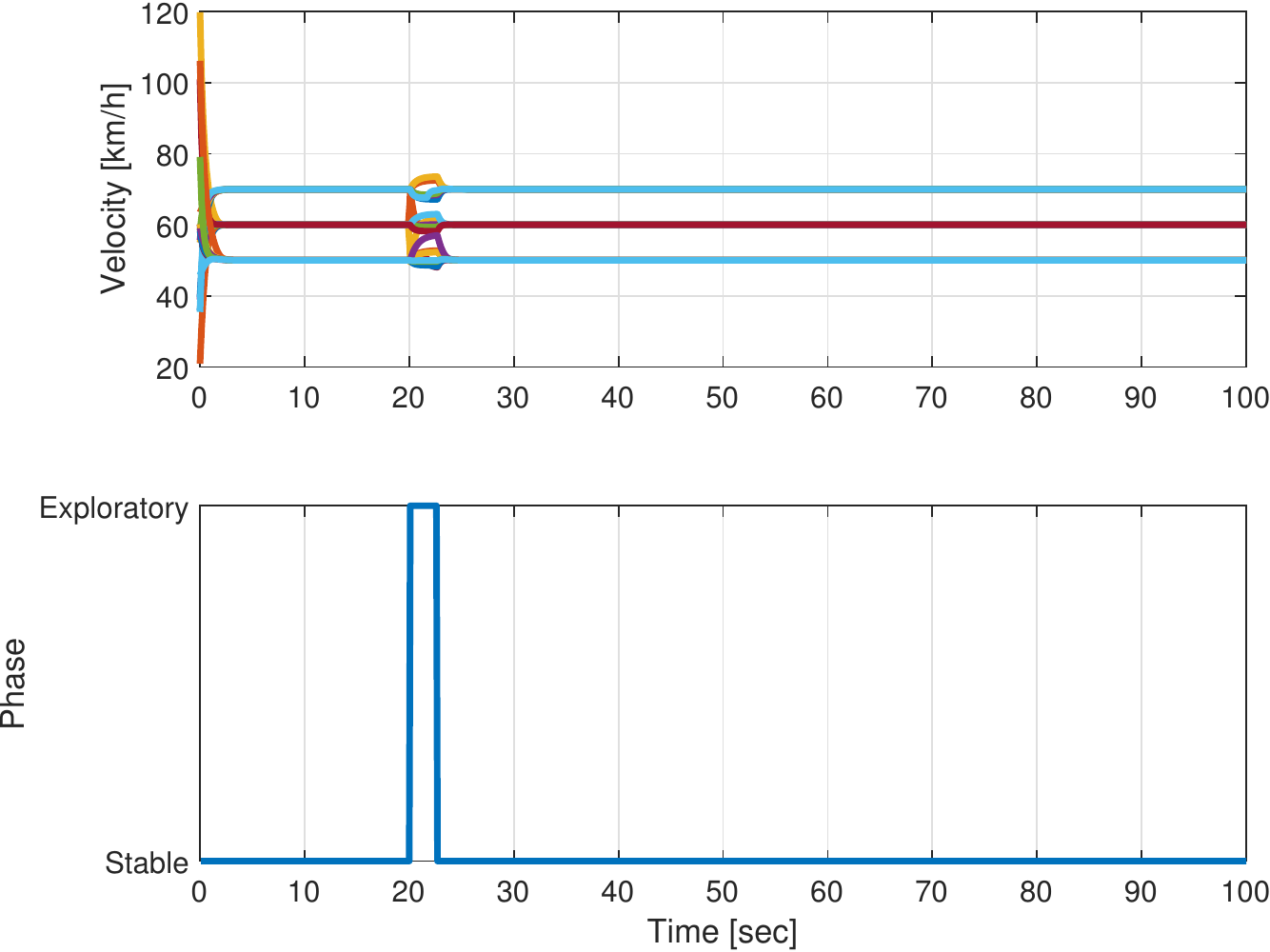}}\label{fig.VehicleExampleCloseFaults}}
	\subfigure[Results of Fourth Scenario]{\scalebox{.5}{\includegraphics{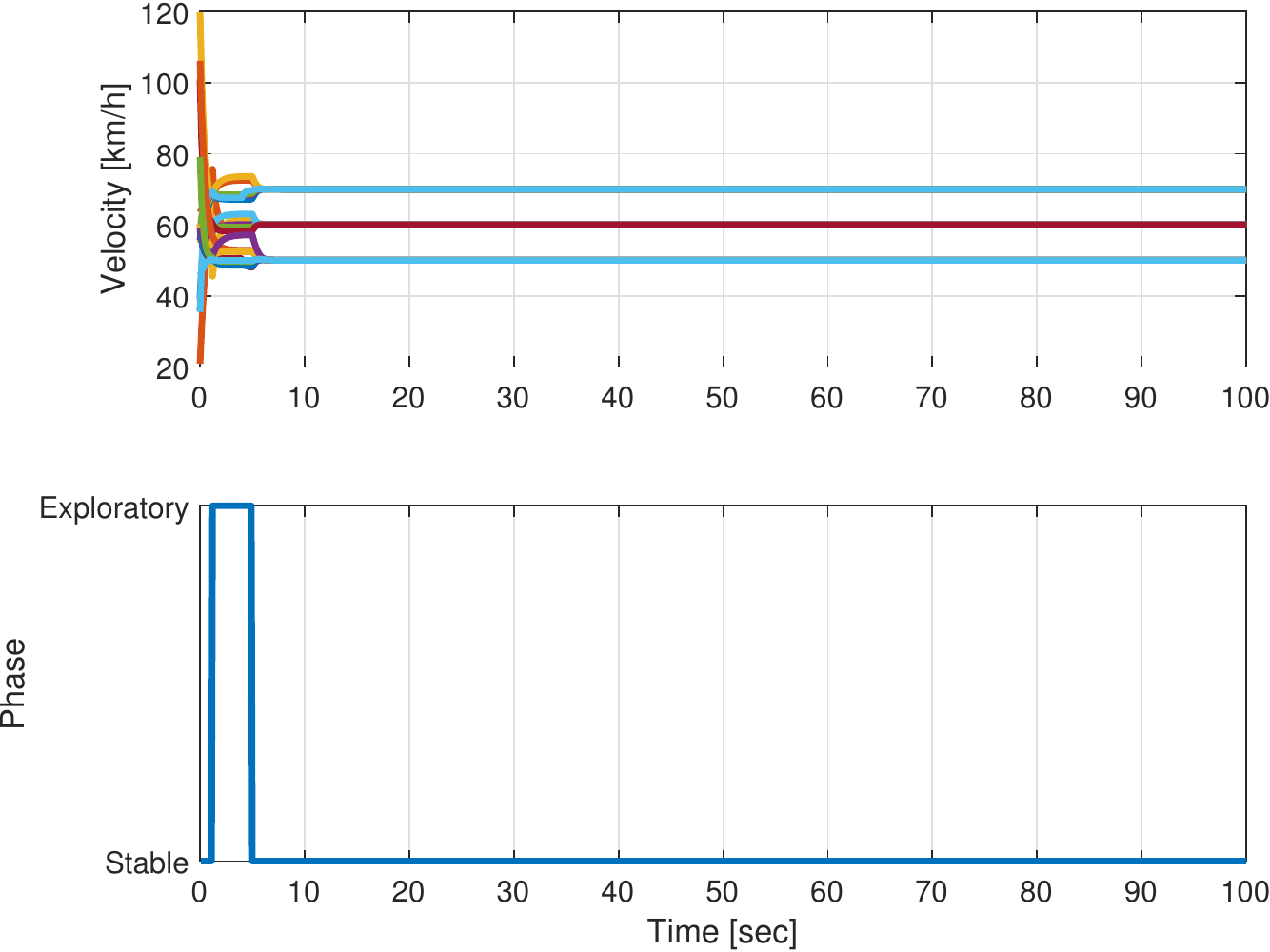}}\label{fig.VehicleExampleCloseEarlyFaults}}
  \caption{Last scenarios in FDI for vehicle systems.}
\end{center}
\end{figure}
\fi
\vspace{-12pt}
\section{Conclusion}
We considered multi-agent networks prone to communication faults in which the agents are {output-strictly MEIP and the controllers are MEIP}. We exhibited a protocol in which a constant bias $\mathrm w$ is added to the controller output, and showed that if $\mathrm w$ is chosen randomly, no matter what the underlying graph $\G$ is, we can asymptotically differentiate between any two versions (faulty or faultless) of the system. We also showed that if $\mathrm w$ is chosen randomly within a certain set, we asymptotically differentiate the faultless version of the system from its faulty version, while also solving the synthesis problem for the faultless version, {assuming $\G$ was connected enough, i.e., $2$-connected}. These results were used to describe algorithms for {network} fault detection and isolation protocols for general MEIP multi-agent systems, where the number of isolable faults is given by a graph-theoretic characteristic of $\G$, while no extra information on the agents and controllers but MEIP is used. These were achieved by assuming the existence of an on-line algorithm asserting that a given network converges to a conjectured steady-state, allowing us to move from asymptotic differentiation to on-line differentiation. Later, two such algorithms were built using passivity of the agents and controllers, {and their correctness was proved}. We demonstrated our protocols by a case study, {in which we successfully detect communication faults in a nonlinear network. We emphasize that the proposed method is proved to work so long the agents and controllers are MEIP, and the graph $\mathcal{G}$ is connected enough. In particular, there is no assumption on the scale of the network.}
Future directions can include more robust {network} fault detection and isolation techniques, in which a larger number of faults can be isolated by studying more delicate graph-theoretical properties of the underlying graph. 
\vspace{-10pt}
\bibliographystyle{ieeetr}
\bibliography{main}

\appendix
This appendix includes the proof of various technical propositions from Subsection \ref{subsec.ConvProf}. We start with the proof of Proposition \ref{prop.Jensen}:
\begin{proof} 
Without loss of generality, we assume that $\Omega_i = \Omega_\star$ for all $i$, as {proving that $\sum_{i=1}^n \rho_i\Omega_\star(\theta_i) \ge C\Omega_\star(\sum_{i=1}^n \theta_i)$ would imply} the desired inequality. We also assume that $\rho_i = 1$ for all $i$, as {proving that $\sum_i \Omega_i(\theta_i) \ge C\cdot\Omega_\star(\sum_i \theta_i)$ would give $\sum_i \rho_i \Omega_i(\theta_i) \ge C\min_i \rho_i \cdot\Omega_\star(\sum_i \theta_i)$}.
Define $F:[0,D]^n\backslash\{0\}\to\mathbb{R}$ as
$
F(\theta_1,\cdots,\theta_n) = \frac{\sum_{i=1}^n \Omega_\star(\theta_i)}{\Omega_\star(\sum_{i=1}^n \theta_i)},
$
where the claim is equivalent to $F$ being bounded from below. For any $r>0$, $F$ is continuous on the compact set $[0,D]^n\backslash\{x:\ ||x|| > r\}$, so its minimum is obtained at some point. As $F$ does not vanish on the set, the minimum is positive, so $F$ is bounded from below on that set by a constant greater than zero. It remains to show that $\lim\inf_{\theta_1,\cdots,\theta_n\to 0} F(\theta_1,\cdots,\theta_n) > 0$.
Let $\beta = \max_i{\beta_i}$, so that $\lim_{\theta\to 0} \frac{\Omega_\star(\theta)}{\theta^\beta} > 0$. Then
$
F(\theta_1,\cdots,\theta_n) = \frac{\sum_{i=1}^n \Omega_\star(\theta_i)}{\Omega_\star(\sum_{i=1}^n \theta_i)}
 = \frac{\sum_{i=1}^n \Omega_\star(\theta_i)}{(\sum_{i=1}^n \theta_i)^\beta} \cdot \frac{(\sum_{i=1}^n \theta_i)^\beta}{\Omega_\star(\sum_{i=1}^n \theta_i)}.
$
We want to bound both factors from below when $\theta_1,\cdots,\theta_n \to 0$. It is clear that the second factor is equal to $(\lim_{\theta\to 0} \frac{\Omega_\star(\theta)}{\theta^\beta})^{-1}$, which is a positive real number by assumption. As for the first factor, we can bound it as $
\lim_{\theta_1\cdots \theta_n \to 0}\frac{\sum_{i=1}^n \Omega_\star(\theta_i)}{(\sum_{i=1}^n \theta_i)^\beta} \ge \lim_{\theta_1\cdots \theta_n \to 0} \frac{\Omega_\star (\max_i \theta_i)}{(n\max_i \theta_i)^\beta} > 0
$
as $\sum_{i=1}^n \theta_i \le n\max_i \theta_i$ and $\sum_i \Omega_\star(\theta_i) \ge \Omega_\star(\max_i \theta_i)$. This completes the proof.
\end{proof}
We now prove Proposition \ref{prop.ConvRateInequality}
\begin{proof}
By assumption $\tilde S(t_k)$ is monotone decreasing and bounded from below, as $\tilde S(t_k) \ge 0$. Thus it converges to some value, denoted $\tilde S_\infty$. Using $\tilde S(t_{k+1}) - \tilde S(t_k) \le -\Omega_\star(\tilde S(t_{k+1}))(t_{k+1}-t_k)$ and taking $k\to \infty$ gives that $0 \le -\Omega_\star(\tilde S_\infty)$. However, $\Omega_\star$ is non-negative, so we must have $\Omega_\star(\tilde S_\infty) =0$, and thus $S_\infty = 0$, meaning that $\tilde S(t_k) \to 0$. By monotonicity of $\tilde S$, we conclude that $\tilde S(t) \to 0$ as $t\to \infty$.
\end{proof}

\end{document}